\documentclass[10pt,journal]{IEEEtran}
\usepackage{graphicx}
\usepackage{subfigure}
\usepackage{graphics}
\usepackage{algorithm}
\usepackage{algorithmic}
\usepackage[algo2e,linesnumbered,ruled]{algorithm2e}
\usepackage{amssymb}
\usepackage{amsmath}
\usepackage{amsthm}
\usepackage[numbers,sort&compress]{natbib}
\usepackage{multirow}
\usepackage{textcomp}
\usepackage{array}
\newtheorem{definition}{Definition}

\newtheorem{thm}{Theorem}

\newcolumntype{M}[1]{>{\centering\arraybackslash}m{#1}}

\ifCLASSOPTIONcompsoc
  \usepackage[nocompress]{cite}
\else
  \usepackage{cite}
\fi
\ifCLASSINFOpdf
\else
\fi
\begin{document}
\title{Probabilistic Top-\textit{k} Dominating Query Monitoring over Multiple Uncertain IoT Data Streams in Edge Computing Environments}

\author{Chuan-Chi Lai,~\IEEEmembership{Member,~IEEE,}
		Tien-Chun Wang,
		Chuan-Ming Liu,~\IEEEmembership{Member,~IEEE,}
		and~Li-Chun Wang,~\IEEEmembership{Fellow,~IEEE}% <-this % stops a space
		\IEEEcompsocitemizethanks{\IEEEcompsocthanksitem C.-C. Lai and L.-C. Wang are with Department of Electrical and Computer Engineering, National Chiao Tung University, 300 Hsinchu, Taiwan (E-mail: cclai1109@nctu.edu.tw; lichun@g2.nctu.edu.tw).
		\IEEEcompsocthanksitem T.-C. Wang is with Department of mobile device development, Compal Electronics, Inc., Taiwan (E-mail: ponyskywangroc@gmail.com).
		\IEEEcompsocthanksitem C.-M. Liu is with Department of Computer Science and Information Engineering, National Taipei University of Technology, 10618 Taipei, Taiwan (E-mail: cmliu@csie.ntut.edu.tw).}
	% <-this % stops an unwanted space
%\thanks{Manuscript received April 19, 2018; revised August 26, 2018.}
}


\markboth{Preprint for IEEE Internet of Things Journal%,~Vol.~XX, No.~X, August~201X
}%
{Lai \MakeLowercase{\textit{et al.}}%: Bare Demo of IEEEtran.cls for Computer Society Journals
}



\IEEEtitleabstractindextext{%
\begin{abstract}
Extracting the valuable features and information in Big Data has become one of the important research issues in Data Science. In most Internet of Things (IoT) applications, the collected data are uncertain and imprecise due to sensor device variations or transmission errors. In addition, the sensing data may change as time evolves. We refer an uncertain data stream as a dataset that has velocity, veracity, and volume properties simultaneously. This paper employs the parallelism in edge computing environments to facilitate the top-\textit{k} dominating query process over multiple uncertain IoT data streams. The challenges of this problem include how to quickly update the result for processing uncertainty and reduce the computation cost as well as provide highly accurate results. By referring to the related existing papers for certain data, we provide an effective probabilistic top-\textit{k} dominating query process on uncertain data streams, which can be parallelized easily. After discussing the properties of the proposed approach, we validate our methods through the complexity analysis and extensive simulated experiments. In comparison with the existing works, the experimental results indicate that our method can improve almost 60\% computation time, reduce nearly 20\% communication cost between servers, and provide highly accurate results in most scenarios.
\end{abstract}

\begin{IEEEkeywords}
Big Data, Internet of Things, Uncertain Data, Multiple Data Streams, Top-\textit{k} Dominating.
\end{IEEEkeywords}}

\maketitle

\IEEEdisplaynontitleabstractindextext

% \fi
\IEEEpeerreviewmaketitle

%\IEEEraisesectionheading{\section{Introduction}\label{sec:introduction}}

\section{Introduction}\label{sec:introduction}

\IEEEPARstart{B}ig data analysis has been widely applied in many fields in recent years. The well-known characteristics of big data are the following Vs: \emph{Volume}, \emph{Velocity}, \emph{Variety}, \emph{Veracity}, \emph{Variability}, and \emph{Value}. Many modern applications and services need to deal with big data from multiple sources. Such a way can be recognized as a computing model over multiple uncertain data streams. For example, some specific applications, \emph{Massive Internet of Things} (Massive IoT)~\citep{8086146}, \emph{Smart City}~\citep{6740844}, and \emph{Location-Based Service} (LBS)~\citep{7929400}, can be recognized as the implementations of a  distributed/parallel sensing data processing model with multiple input uncertain data streams. The afore mentioned applications match at least three big data's characteristics: volume, velocity, and veracity. The volume of information is growing all the time so that an efficient parallel or distributed computing way is required. In massive IoT environments, the real-time monitoring is a typical application for detecting the events that need to be avoided or alleviated. In this case, the users/operators only concern the latest results for most queries and thus the information has time-limited (or velocity) feature. Due to the unreliability of data retrieval process, many data are inaccurate or uncertain. In such a case, the probabilities are used to represent the distribution of different situations. Hence, a massive IoT application has to effectively process the uncertain data to provide near real-time results with high precision (or veracity).

Although the big data can be resolved by the \emph{Cloud Computing} model, the response time (or latency) still can not meet the requirements of some near real-time IoT monitoring applications. \emph{Edge Computing}~\citep{8166730,8089336} thus has become the promising architecture to improve the response time for IoT applications in recent years. Most researchers focus on developing new techniques to edge computing from  system design, communications, networking, and resource management~\citep{GAI201646,GAI2018126,8425298} aspects. However, developing new effective techniques to process the IoT data efficiently from data science/engineering aspects is also very important and helpful to the IoT applications. Many researchers thus have proposed some algorithms for different types of queries (demands) to find the insightful knowledge in big data and make the precise decision. \emph{Skyline}~\citep{Hose2012,LIU201540,1583586} and \emph{Top-$k$}~\citep{7095576,Mouratidis:2006:CMT:1142473.1142544,7831369} queries are common research topics. However, the skyline and top-$k$ queries lead to some discrepancies in the search results. Nowadays, such two queries cannot satisfy the demand of some modern applications. Therefore, a new query, \emph{Top-$k$ Dominating}~\citep{7166329,7845627,8248785,8267252}, in certain data combined the above two search features comes into being.

In general, an uncertain data object is usually modeled with multiple probabilities which represents the probabilities of the object's occurrences or errors for some applications, such as IoT data analysis. Such a data model makes the query process much more complicated.
Some works~\citep{Zhang:2010:TPT:1773653.1773676,10.1007/978-3-319-05810-8_26,10.1007/978-3-319-11749-2_19} have discussed the \emph{Probabilistic Top-$k$ Dominating} (PTD) query processing on uncertain data streams.
In traditional, to handle probabilistic top-$k$ dominating queries, the system will compute the \emph{dominant scores} between different data objects and find out $k$ objects having the highest dominant scores. Such a straightforward process needs $O((n|U|)^2)$ computation time, where $n$ is the number of instances in an object and $U$ is the input data set. As the amount of data increases dramatically, the system needs solutions to effectively reduce the computational complexity. In the traditional centralized systems, R-tree~\citep{Guttman:1984:RDI:602259.602266} is one of the most popular indexing structure to improve the performance of query processing. Due to the spatial characteristic of R-tree, the system can get a great performance improvement on following operations: object search, value comparison, and pruning. However, utilizing centralized data structures and algorithms already can not handle the big data lead by the IoT era. Therefore, it is reasonable to improve the efficiency of computations using modern parallel and distributed computations. In this paper, we propose a \emph{Probabilistic Top-$k$ Dominating query process over Multiple Uncertain data Streams} (PTDMUS) algorithm to improve the efficiency of searching $k$ data objects that have the highest dominate scores for the distributed real-time IoT monitoring applications. The contributions of this work are listed as follows.
\begin{itemize}
	%\item To the best of our knowledge, the proposed approach, PTDMUS, is the first work to discuss the probabilistic top-$k$ dominating (PTD) query monitoring over uncertain data streams in parallel/distributed edge computing environments for IoT real-time monitoring applications.
	\item We provide a parallel processing model utilizing the R-trees, $k$-skyband~\citep{Mouratidis:2006:CMT:1142473.1142544}, and a threshold for effectively precluding irrelevant objects in advance, and thereby significantly reduce the computational overhead. Such an idea can be used to solving some other similar types of queries.
	\item We propose an estimated theorem for the distributed computing environments to effectively predict the time that a data object has the chance to become the final result, and consequently decreases the frequency of dominance checks on the edge computing nodes.
	\item In addition, we present the theoretical analysis of PTDMUS on time complexity, space complexity, and transmission cost in the average and worst cases.
	\item The simulation result indicates that PTDMUS outperforms the conventional method with 60\% computation time and 20\% transmission cost while keeping near 100\% precision and recall of final query result in most scenarios.
\end{itemize}

The rest of paper is organized as follows. Some related researches are reviewed in Section~\ref{sec:relatedwork}. Section~\ref{sec:preliminaries} presents the definitions, notations, and problem statement of this work. Section~\ref{sec:ptdmus} discusses the proposed solutions with some algorithms and running examples in details.
Some theoretical analysis and discussion are explained in Section~\ref{sec:analysis}.
Simulation results are presented in Section~\ref{sec:simulation}. Finally, we give concluding remarks in Section~\ref{sec:conclusion}.

\vspace{-15pt}
\section{Related Work}
\label{sec:relatedwork}
Many researchers have discussed range, skyline, and top-$k$ queries over uncertain data in distributed computing environments. Nowadays, the above query types can not satisfy the demand of some modern applications.
Hence, we focus on a more complex query, top-$k$ dominating query, in this work. In the balance of this section, we introduce the related works about top-$k$ dominating query processing from the data science aspect. The comparisons of conventional works are summarized in Table~\ref{compared_methods} and each work will be described in followings.

Miao et al.~\citep{7166329} proposed a Bitmap Indexing Guided (BIG) algorithm for improving the performance of processing top-$k$ dominating query on large incomplete dataset. Han et al.~\citep{7845627} provided a table-scan-based method with presorted results for improving the performance/efficiency of top-$k$ dominating query computations on massive data in batch computing model.
Amagata et al.~\citep{8248785} mapped multiple input datasets into a data space and then proposed a method which generates virtual points for effectively precluding unnecessary data objects in the data space.
Ezatpoor et al.~\citep{8267252} applied BIG algorithm~\citep{7166329} to MapReduce framework for providing a parallel computing model to enhance the performance of processing top-$k$ dominating query on large incomplete dataset. However, only~\citep{8248785} and~\citep{8267252} proposed the algorithms for distributed computing environments. Furthermore, the above approaches for certain data did not support continuous query processing in real-time IoT monitoring applications.

For uncertain data, only few studies~\citep{10.1007/978-3-319-11749-2_19,10.1007/978-3-319-05810-8_26,Zhang:2010:TPT:1773653.1773676} have explored the top-$k$ dominating query processing until now. Zhang et al.~\citep{Zhang:2010:TPT:1773653.1773676} proposed a threshold-based algorithm to prune the irrelevant objects and thus improved the performance of computation for top-$k$ dominating query. Zhan et al.~\citep{10.1007/978-3-319-05810-8_26} developed new pruning techniques by utilizing the spatial indexing and statistic information while considering the maximum/upper and minimum/lower bounds of probabilistic dominance, which reduced computational and I/O costs. Li et al.~\citep{10.1007/978-3-319-11749-2_19} proposed a method to postpone the unnecessary calculation if the query results did not change dramatically in a certain period of time and the computational cost could be reduced. However, these works did not consider how to process continuous queries over uncertain data with parallelisms for real-time IoT monitoring applications based on edge computing environments.

In summary, to the best of our knowledge, none of existing works simultaneously consider following characteristics: uncertain data, continuous probabilistic top-$k$ dominating query, distributed computing, and the real-time requirement for IoT monitoring. This shows that probabilistic top-$k$ dominating query processing over uncertain data for edge-enabled IoT real-time monitoring applications remains a big challenge.

\begin{table}[ht]
	\renewcommand{\arraystretch}{1.2}
	\caption{Comparisons of Related Works and the Proposed Method}
	\label{compared_methods}
	\centering
	%\small
	\begin{tabular}{lM{1.3cm}M{1.35cm}M{1.4cm}M{1.2cm}}
		\hline
		& \multicolumn{4}{c}{\textbf{Characteristics}} \\
		\cline{2-5}
		\textbf{Methods} & \textbf{Data Type} & \textbf{Continuous Query} & \textbf{Distributed Computing} & \textbf{Real-time} \\
		\hline
		\hline
		BIG~\citep{7166329} & Certain & \texttimes & \texttimes & \texttimes \\
		TDTS~\citep{7845627} & Certain & \texttimes & \texttimes & \texttimes \\
		SFA~\citep{8248785} & Certain & \texttimes & \checkmark & \texttimes \\
		MRBIG~\citep{8267252} & Certain & \texttimes & \checkmark & \texttimes \\
		TPTD~\citep{Zhang:2010:TPT:1773653.1773676} & Uncertain & \texttimes & \texttimes & \texttimes \\
		PTOPK~\citep{10.1007/978-3-319-05810-8_26} & Uncertain & \texttimes & \texttimes & \texttimes \\
		PEA~\citep{10.1007/978-3-319-11749-2_19} & Uncertain & \checkmark & \texttimes & \checkmark \\		
		\textbf{PTDMUS} & Uncertain & \checkmark & \checkmark & \checkmark \\
		\hline
	\end{tabular}
\end{table}

\section{Preliminaries}
\label{sec:preliminaries}
In this section, we introduce the fundamental assumptions, the system model, and the problem statement.
%Table~\ref{notations} summarizes all the symbols used in the rest of this paper.
%\begin{table}[ht]
	%\renewcommand{\arraystretch}{1.2}
%	\caption{Notation}
%	\label{notations}
%	\centering
%	\begin{tabular}{lp{6.65cm}}
%		\hline
%		\textbf{Symbol} & \textbf{Description}\\
%		\hline
%		$U$ & a set of $d$-dimensional uncertain data objects\\
%		$u$ & a $d$-dimensional uncertain data object with probabilistic values on all the dimensions\\
%		$u^{a}$ & an instance of $u$\\
%		$u^{a}[d]$ & the $d$th attribute value of instance $u^{a}$\\
%		$N_H$ & the cloud/header/coordinator node\\
%		$N_j$ & an edge/monitor node\\	
%		$US_j$ & an uncertain data stream with respect to $N_j$\\
%		$SW_H$ & the global sliding window in $N_H$\\
%		$SW_j$ & the local sliding window in $N_j$\\		
%		$R$ & an R-tree\\
%		$\prec$ & dominance\\	
%		$PTD$ & the set of final probabilistic top-$k$ dominating objects\\
%		$KS$ & the set of uncertain data objects in $k$-skyband\\
%		$CS$ & the candidate set for a query\\
%		$CT$ & the checking-time table recording the  next checking time of each candidate object\\
%		$t$ & the time slot/stamp\\
%		$dom(u)$ & the dominant score of uncertain object $u$\\
%		$r$-$dom(u)$ & the dominated score of uncertain object $u$\\
%		$dom_k$ & the $k$-th highest dominant score of the objects in $PTD$\\
%		\hline
%	\end{tabular}
%\end{table}
%

\subsection{Fundamental Assumptions}
Three kinds of uncertain data models have been proposed and discussed in~\citep{Wang2013}: fuzzy, evidence-oriented, and probabilistic models. In this work, we refer to the last model with discrete case and the uncertain data object can be defined as Definition~\ref{def_udo}.
%\vspace{.3em}
\begin{definition}[\textbf{Uncertain Data Objects}]
	\label{def_udo}
	Given a $d$-dimensional uncertain data set $U$, each uncertain data object $u\in U$ with $n$ instances is a probability distribution over the $d$-dimensional space. Each instance $u^{a}$ of $u$ has $d$ attributes, $u^{a}[1], u^{a}[2], \cdots, u^{a}[d]$, where $a=1,\dots, n$, and is associated with a probability $Pr(u^{a})$, where $Pr(u)=\sum_{a=1}^{n}Pr(u^{a})=1$.
\end{definition}

A simple example of a two-dimensional uncertain data set is presented in Table~\ref{2d_uncertain_data_set}, in which each uncertain data object has three possible instances. For example, object $u_1$ has three instances $u_1^1$, $u_1^2$, and $u_1^3$ with probabilities $0.4$, $0.3$, and $0.3$, respectively. It means that $u_1$ may occur in three possible cases with different corresponding probabilities and the total probability of all cases will be 1. Note that we will use attribute or dimension interchangeably.
%\vspace{-5pt}
%
\begin{table}[ht]
	\centering
	\caption{An example of a two-dimensional uncertain data set}
	\label{2d_uncertain_data_set}
	\renewcommand{\arraystretch}{1.3}
	\begin{tabular}{|c|c|c|c|}
		\hline
		\textbf{Object} & \textbf{Instance} & \textbf{Object} & \textbf{Instance} \\ \hline
		\multirow{3}{*}{$u_1$} & $u_1^1[0.4,28,7]$ & \multirow{3}{*}{$u_2$} & $u_2^1[0.6,21,16]$ \\ \cline{2-2} \cline{4-4}
		& $u_1^2[0.3,31,11]$ & & $u_2^2[0.1,17,21]$ \\ \cline{2-2} \cline{4-4}
		& $u_1^3[0.3,35,8]$ & & $u_2^3[0.3,15,17]$ \\ \hline
		\multirow{3}{*}{$u_3$} & $u_3^1[0.7,72,33]$ & \multirow{3}{*}{$u_4$} & $u_4^1[0.8,48,19]$ \\ \cline{2-2} \cline{4-4}
		& $u_3^2[0.2,67,30]$ & & $u_4^2[0.1,43,23]$ \\ \cline{2-2} \cline{4-4}
		& $u_3^3[0.1,64,35]$ & & $u_4^3[0.1,52,26]$ \\ \hline
	\end{tabular}
\end{table}

If we map the instances of the uncertain data objects onto a $d$-dimensional space, each uncertain data object $u$ can be represented by a \emph{minimum bounding rectangle}, MBR($u$), which is the minimum rectangle containing all the instances of $u$ in the space. Let $u^{\max}$ and $u^{\min}$ respectively denote the maximum and minimum corners of $u$ where $u^{\max} [\alpha]=\max_{1\leq a \leq n} {u^{a}[\alpha]}$ and $u^{\min} [\alpha]=\min_{1\leq a \leq n} {u^{a}[\alpha]}$, where $\alpha=1, \ldots,d$. Then, MBR($u$) can be represented by $[u^{\min}, u^{\max}]$ where $u^{\min}=(u^{\min}[1],u^{\min}[2],\dots,u^{\min}[d])$ and $u^{\max}=(u^{\max}[1],u^{\max}[2],\dots,u^{\max}[d])$.
Note that $u^a[0]$ is the probability value $Pr(u^a)$ of instance $u^a$.
According to the example in Table~\ref{2d_uncertain_data_set}, $u_1^{\max}[1]=35, u_1^{\min}[1]=28, u_1^{\max}[2]=11,$ and $u_1^{\min}[2]=7$, so MBR($u_1$) is $[u_1^{\min}, u_1^{\max}] = [(28,7),(35,11)]$. Fig.~\ref{ex:r-tree} shows the MBRs of each data objects indexed by an R-tree for the example in Table~\ref{2d_uncertain_data_set}. The four uncertain data objects $u_1$, $u_2$, $u_3$, and $u_4$, on a 2D plane with the associated MBRs and each object has three instances respectively. Note that we use the bulk loading algorithm~\citep{DBLP:conf/caise/LeeL03} to construct the R-trees~\citep{Guttman:1984:RDI:602259.602266} in our work since it can utilize the space and avoid the overlapping issue between MBRs, thus improving the query time.
%\vspace{.3em}
\begin{figure}[ht]
	%\vspace{-10pt}
	\centering
	\includegraphics[width=0.5 \textwidth]{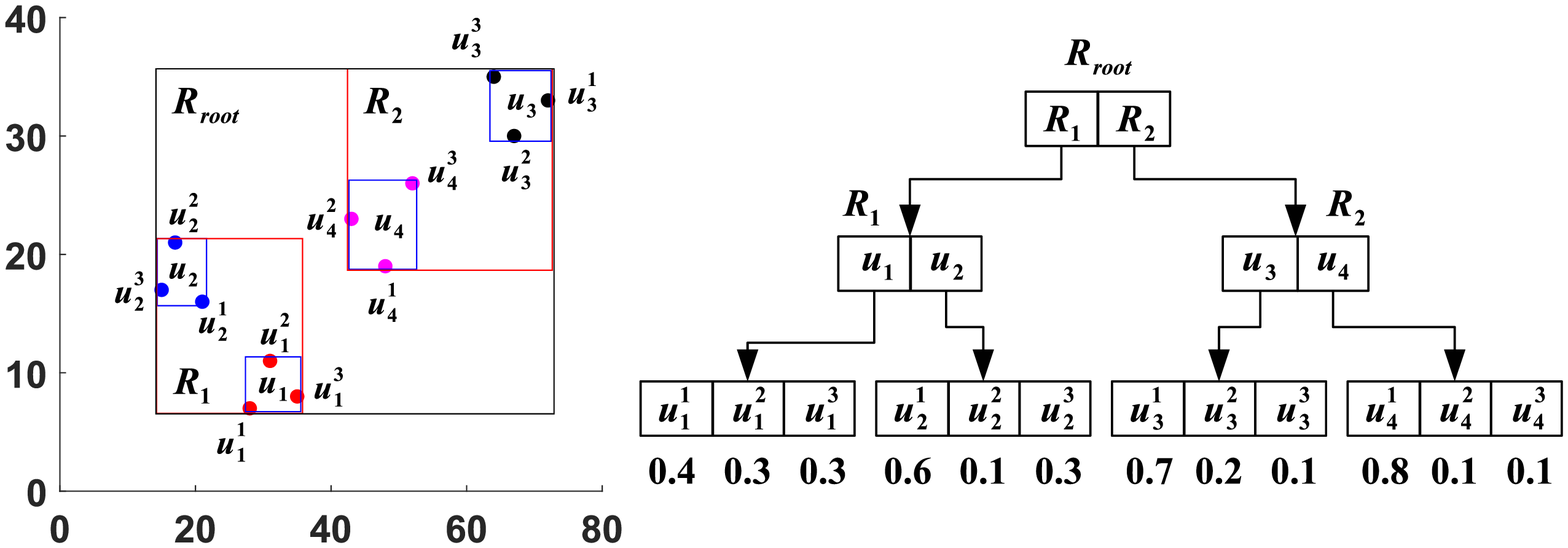}
	\caption{Example of an R-tree with four data objects (MBRs) from the data set in Table~\ref{2d_uncertain_data_set}.}
	\label{ex:r-tree}
	%\vspace{-10pt}
\end{figure}

We consider multiple uncertain data streams and denote an Uncertain data Stream as \emph{US}, where the uncertain data objects are generated with time and will be invalid after a period of time. Data streams play an important role in the era of big data with the advance of IoT technology and have attracted much attention for years. Most of the related researches use the sliding window model and focus on the recent data in the stream. Our work follows this trend.

In a data stream, each data object has a time stamp to denote the time for entering the system. This scenario is usually modeled as the sliding window and it can be defined as below.
%\vspace{.3em}
\begin{definition}[\textbf{Sliding Window}]
	\label{def_sw}
	Suppose the sliding window, \emph{SW}, is of size $|SW|$. Then the data newly generated will be invalid after $|SW|$ time instances. We use $SW[t-|SW|+1,t]$ to denote the set of the uncertain data objects in the current sliding window at time $t$. The considered sliding window follows the first-in-first-out rule for keeping the objects.
\end{definition}
%\vspace{.3em}
In this paper, we use $u_1, u_2, \dots, u_{|\mbox{\emph{SW}}|}$ to denote the uncertain data objects in \emph{SW} (i.e., \emph{SW}= $\{u_1, u_2, \dots, u_{|\mbox{\emph{SW}}|}\}$, according to the arrival time of each data object. To search the top-$k$ dominating objects, the system will use the dominant scores obtained from the skyline query. In our work, we assume that the value of a dimensional attribute is the smaller, the better. The dominant relations between instances can be defined as follows.
%\vspace{.3em}
\begin{definition}[\textbf{Instance-level Dominance}]
	\label{def_IlD}
	Given two uncertain data instances of two different uncertain data objects $u_i^{a}$ and $u_j^{b}$ where $a, b\in [1,n]$ and $i\ne j$, if the condition $(\forall \alpha\in[1,d], u_i^{a}[\alpha]\leq u_j^{{b}}[\alpha])\wedge(\exists \beta\in[1,d], u_i^{a}[\beta]< u_j^{b}[\beta])$ holds, we say $u_i^{a}$ dominates $u_j^{b}$ and it is denoted as $u_i^{a}\prec u_j^{b}$.
\end{definition}
%\vspace{.3em}

In short, none of $u_j^{b}$'s attributes is better (smaller and except for equal) than $u_i^{a}$'s the corresponding dimensional attribute. Since a data object may has multiple instances, we can classify the object-level dominance into three cases. The relevant definitions are presented in the following.
%\vspace{.3em}
\begin{definition}[\textbf{Object-level Dominance}]
	\label{def_OlD}
	Suppose there are two uncertain data objects $u_i$ and $u_j$ and each object has $n$ instances. If $u_i$ is considered as a dominator, the relation between $u_i$ and $u_j$ can be classified by using following cases:
	\begin{enumerate}
		\item Complete Dominance: all the instances of $u_i$ dominate all the instances of $u_j$, denoted as $u_i \prec u_j$.
		\item Partial Dominance: some instances of $u_i$ dominate some instances of $u_j$, denoted as $u_i \precsim u_j$.
		\item Missing Dominance: no instances of $u_i$ dominate any instance of $u_j$, denoted as $u_i \nprec u_j$.
	\end{enumerate}	
	In summary, the probability of $u_i$ dominating $u_j$ can be generally expressed as
	\begin{equation}\label{eq_old}
	Pr[u_i\prec u_j]=\sum_{a=1}^{n}(Pr(u_i^a)\times \sum_{\forall u_j^b\in u_j, u_i^a\prec u_j^b}Pr(u_j^b)).\nonumber
	\end{equation}
\end{definition}
%\vspace{.3em}

According to the above definitions, we can derive the score of the dominant relation between two objects according to the following definition.
%\vspace{.3em}
\begin{definition}[\textbf{Dominant Score of an Object}]
	\label{def_dpo}
	Given an uncertain data object $u_i$ with $n$ instances, the expected dominant score of an instance $u_i^{a}$ can be derived by
	\begin{equation}\label{eq_dpi}
	dom(u_i^{a})=\sum_{u_j\in U,i\neq j} \{Pr(u_i^{a})\times Pr(u_j^{b})|u_i^{a}\prec u_j^{b}\}.\nonumber
	\end{equation}
	Then, the dominant score of the uncertain object $u_i$ is defined as
	\begin{equation}\label{eq_dpo}
	dom(u_i)=\sum_{a=1}^{n}dom(u_i^{a}).\nonumber
	\end{equation}
\end{definition}
%\vspace{.3em}

Consider the example in Table~\ref{2d_uncertain_data_set} and Fig.~\ref{ex:r-tree}, where instances $u_1^1$, $u_1^2$, and $u_1^3$ dominate the following instances: $u_3^1$, $u_3^2$, $u_3^3$, $u_4^1$, $u_4^2$, and $u_4^3$ by Definition~\ref{def_IlD}. In other words, by Definition~\ref{def_OlD}, object $u_1$ completely dominates objects $u_3$ and $u_4$, denoted as $u_1 \prec u_3$ and $u_1 \prec u_4$. The derivation of $dom(u_1)$ can be presented as
\begin{align*}
dom(u_1)=&\sum_{a=1}^{3}dom(u_1^{a})\\
=&[Pr(u_1^1)+Pr(u_1^2)+Pr(u_1^3)] \\
\times&[Pr(u_3^1)+Pr(u_3^2)+Pr(u_3^3)+Pr(u_4^1)\\
+&Pr(u_4^2)+Pr(u_4^3)]\\
=&Pr(u_1)\times [Pr(u_3)+Pr(u_4)]\\
=&1\times (1+1)=2.
\end{align*}
For object $u_2$, it completely dominates object $u_3$ and partially dominates object $u_4$, so the calculation of $dom(u_2)$ will be
\begin{align*}
dom(u_2)=&Pr[u_2\prec u_3]+Pr[u_2 \precsim u_4] \\
=&1+(Pr[u_2^1\prec u_4]+Pr[u_2^3\prec u_4]+Pr[u_2^2\precsim u_4]) \\
=&1+[Pr(u_2^1)+Pr(u_2^3)]\times Pr(u_4)\\
+&Pr(u_2^2)\times[Pr(u_4^2)+Pr(u_4^3)] \\
=&1+(0.6+0.3)\times 1+0.1\times(0.1+0.1) \\
=&1+0.9+0.02=1.92.
\end{align*}
Consequently, we can obtain the dominant scores of all the uncertain objects in the same way.

\subsection{System Architecture}
In this work, we construct an edge computing system as shown in Fig.\ref{ex:system_model} and make it support the parallel and distributed computing for monitoring the top-$k$ query over multiple uncertain IoT data streams. Such a way can improve the efficiency of computation. The system consists of a coordinator node (cloud service) $N_H$ and $m$ monitor nodes (edge computing nodes) $N_1,N_2, \dots,N_m$. Each monitor node $N_j$ can directly contact with the coordinator node $N_H$, where $1\leq j\leq m$. For $N_H$, all the reported information from each $N_j$ is recognized as an uncertain data steam $US_j$. Each $N_j$ needs to continuously compute the local result of the query and upload it to $N_H$ as the candidate result. $N_H$ needs to record all the unexpired candidates that are received from each $N_j$.
%Each $N_j$ uses an R-tree to index the all the unexpired data objects in the stream $US_j$.

\subsection{Problem Statement}
\label{sec:def_ptdmus}
Our objective is to have a time-efficient approach determining the uncertain objects with top-$k$ dominant scores among all the uncertain objects in the considered system model. A global sliding window $SW_H$ and $m$ uncertain data streams $US_1, US_2, \ldots, US_m$ are given. Each $US_j$ is corresponding to the monitor node $N_j$ where $1\leq j\leq m$. Each $N_j$ has its local $SW_j$ and $|SW_H|=m*|SW_j|$. Each $N_j$ examines the objects in $SW_j$, saves the possible objects in a local candidate set, and then reports the local candidate set to the coordinator node $N_H$. $N_H$ uses the received local candidate sets to calculate the global candidate set and then broadcasts it to each $N_j$. Each $N_j$ uses the received global candidate set to derive the dominant scores of the objects that dominate others and then returns the scores to $N_H$. After that, $N_H$ integrates the received score information of each object and finds out $k$ data objects that have the highest scores. The final result set including \emph{Probabilistic Top-$k$ Dominating objects} is denoted as $PTD$. Note that the above process are repeatedly operated until there is no input data.

According to the above assumptions, there are three important issues to be solved:
\begin{enumerate}
	\item How to avoid the unnecessary computations for the dominant scores in order to save the computation time?
	\item How to minimize the number of local candidate objects for improving the transmission cost?
	\item How to reduce the frequency of dominant score derivations as time evolves?
\end{enumerate}

\begin{figure*}[t]
	\centering
	\includegraphics[width=0.7 \textwidth]{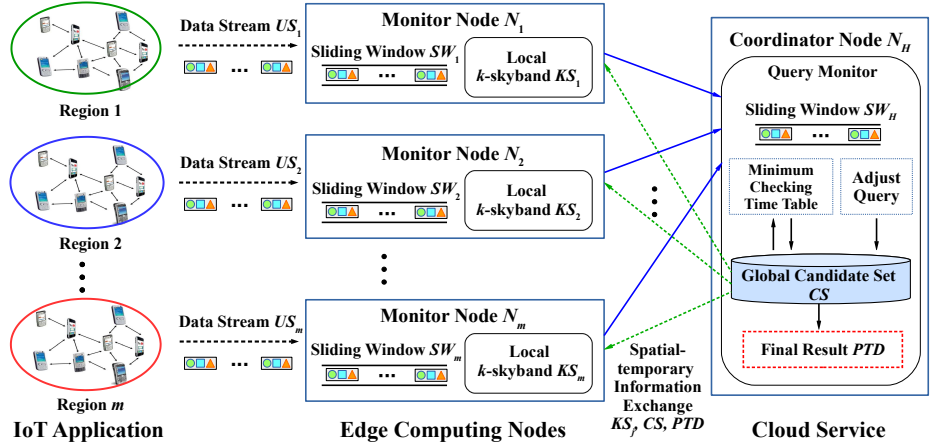}
	\caption{The architecture of considered edge computing system for query monitoring.}
	\label{ex:system_model}
	\vspace{-10pt}
\end{figure*}

\section{Probabilistic Top-\textit{k} Dominating Query Process over Multiple Uncertain Data Streams (PTDMUS)}
\label{sec:ptdmus}
In this section, we present the proposed approach, Probabilistic Top-\textit{k} Dominating Query Process over Multiple Uncertain Data Streams (PTDMUS). PTDMUS provides three mechanisms to solve the above three issues and we respectively introduce each of them in detail.

\subsection{The Computation with R-trees}
In the first part, we apply R-trees to the considered system for improving the computational speed of the dominant score derivation. Note that we use~\citep{GarciaR:1998:GAB:288692.288723} to generate bulk loading R-trees and minimize the overlaps of elements in each level, thereby optimizing the searching time.
By combining the characteristics of an MBR in Definition~\ref{def_OlD}, we can define the dominant relation between different MBRs as Definition~\ref{def_DMBR}.
%\vspace{.3em}
\begin{definition}[\textbf{Dominance between different MBRs}]
	\label{def_DMBR}
	Given two different minimum bounded rectangles $MBR(u_1)$, $MBR(u_2)$ of object $u_1$ and $u_2$ respectively, if we consider $MBR(u_1)$ as a dominator, we can classify the relation between $MBR(u_1)$ and $MBR(u_2)$ as following cases:
	\begin{enumerate}
		\item Complete Dominance: $u_1^{max}$ of $MBR(u_1)$ is smaller than of $u_2^{min}$ of $MBR(u_2)$, denoted as $MBR(u_1)\prec MBR(u_2)$.
		\item Partial Dominance: $u_1^{min}$ of $MBR(u_1)$ is smaller than of $u_2^{max}$ of $MBR(u_2)$, denoted as $MBR(u_1)\precsim MBR(u_2)$.
		\item Missing Dominance: $u_1^{min}$ of $MBR(u_1)$ is larger than of $u_2^{max}$ of $MBR(u_2)$, denoted as $MBR(u_1)\nprec MBR(u_2)$.
	\end{enumerate}	
\end{definition}
%\vspace{.3em}
Using the cases in Definition~\ref{def_DMBR}, the system can preclude irrelevant objects effectively and thus improve the computation overhead for dominant scores. For example, to derive the dominant score of $u_i$, the system uses MBR($u_i$) and the R-tree as inputs. The system will put all the children of the root in a \emph{Target Set} ($TS$) and then examine the relation between MBR($u_i$) and each the element $e_k$ in $TS$:
\begin{enumerate}
	\item Complete Dominance: if $e_k$ is an MBR node, add the number of objects in $e_k$ to $dom(u_i)$; otherwise, $e_k$ is an object and $dom(u_i)=dom(u_i)+1$.
	\item Partial Dominance: if $e_k$ is an MBR node, put all the children of $e_k$ in the \emph{Next Target Set} ($NTS$); otherwise, $e_k$ is an object and $dom(u_i)$ directly is added to the dominant score of $u_i$ with respect to $e_k$.
	\item Missing Dominance: do nothing.
\end{enumerate}	
After examining all the elements in $TS$, if $NTS$ is not empty, the system will clear $TS$ and insert all the elements of $NTS$ to $TS$. The system will do the above operations repeatedly until $TS$ is empty and can obtain the dominant scores of $u_i$ with respect to all the objects in the R-tree. The above computation process with R-tree will be executed on monitor nodes in PTDMUS.

\subsection{Threshold-based Probabilistic \textit{k}-skyband}
To reduce the number of candidate objects,~\citep{Mouratidis:2006:CMT:1142473.1142544} proposed a $k$-skyband approach for the top-$k$ query on certain data. In our work, we follow this idea to define a probabilistic $k$-skyband for minimizing the size of candidate set. First, we define the \emph{dominated score} of an uncertain object as follows.
%\vspace{.3em}
\begin{definition}
	\label{def_POD}
	Given an uncertain data object $u_i$ in the sliding window $SW$, the score of an instance $u_i^{a}$ being dominated (or called dominated score) is defined as
	\begin{equation}
	r\text{-}dom(u_i^{a})=\sum_{u_j^{b}\prec u_i^{a},u_j\in SW-\{u_i\}}Pr(u_i^{a})\times Pr(u_j^{b}),\nonumber
	\end{equation}		
	where $u_j^{b}$ is an instance of $u_j$, and the dominated score of $u_i$ is	\begin{equation}
	r\text{-}dom(u_i)=\sum_{a=1}^n r\text{-}dom(u_i^{a}).\nonumber
	\end{equation}
\end{definition}
%\vspace{.3em}
After obtaining the dominated score of each object, we can define the probabilistic $k$-skyband ($KS$) as follows.
%\vspace{.3em}
\begin{definition}[\textbf{Probabilistic \textit{k}-skyband}]
	\label{def_PKOD}
	For a given integer $k$, the probabilistic $k$-skyband ($KS$) is a set of uncertain data objects and $KS=\{u \in U|dom(u)\geq 1\wedge r\text{-}dom(u) < k\}$
\end{definition}
%\vspace{.3em}

Note that the top-$k$ dominating result is always a part of the $k$-skyband~\citep{Papadias:2005:PSC:1061318.1061320} in certain data.
For uncertain data, we can also have a similar property as shown in Theorem~\ref{def_TPKOD}
%After observing some small running cases with such a similar idea, we find an important theorem as Theorem~\ref{def_TPKOD}.
%\vspace{.3em}
\begin{thm}
	\label{def_TPKOD}
	Given a set of probabilistic top-$k$ dominating objects $PTD$, $u\in PTD$, if $dom(u)\geq 1$, then $u\in KS$.
\end{thm}
%\vspace{.3em}
\begin{proof}	
	If $u \notin KS$, then $r$-$dom(u)\geq k$ according to Definition~\ref{def_PKOD}. In this case, at least $k$ other objects dominate $u$ in average. It is hence impossible for $u$ to be one of the top-$k$ dominating objects and $u\notin PTD$. This contradicts to the given condition and the proof is done.
\end{proof}

In order to use the property in Theorem~\ref{def_TPKOD} for processing probabilistic top-$k$ dominating queries in parallel, we derive \emph{Threshold-based Probabilistic $k$-skyband} by giving a threshold $\delta \leq k$.
%\vspace{.2em}
\begin{definition}[\textbf{Threshold-based Probabilistic \textit{k}-skyband}]
	\label{def_TKS}
	For a given integer $k$ and a threshold value $\delta \leq k$, the threshold-base probabilistic $k$-skyband (TKS) is a subset of $KS$ and $TKS=\{u\in U|dom(u)\geq 1 \wedge r\text{-}dom(u)< \delta \leq k\}$
\end{definition}
%\vspace{.3em}

This method is used to solve the second issue we mentioned in the problem statement and reduce the size of the candidate set in each monitor node. In this way, the coordinator node can also process less received candidate objects that are possible to be the top-$k$ dominating objects. In conventional methods~\citep{Papadias:2005:PSC:1061318.1061320,Amagata2016}, the $k$-skyband is computed on both monitor and coordinator nodes. However, in the most modern big data applications, such a way is not efficient since the volume of candidates are usually still too large for the top-$k$ dominating objects considered by users. The coordinator node still needs too much computational cost on the $k$-skyband calculation and makes the response time intolerable to users. Hence, in the proposed PTDMUS approach the coordinator node uses a new mechanism, \emph{Minimum Checking Time} (MCT), to help efficiently derive the final result instead of computing the global threshold-based probabilistic $k$-skyband as the candidate set, $CS$. Note that $CS=\bigcup_{j=1}^{m}TKS_j$ and $TKS_j$ is the local result of threshold-based probabilistic $k$-skyband from the monitor node $N_j$. Such a mechanism can help the coordinator node process the objects that are really relevant to the query, decrease the frequency of score derivation on irrelevant objects, and significantly reduce the computational cost as well as improve the response time.

\subsection{Minimum Checking Time}
We discover an important phenomenon that each monitor node usually uploads the candidate set that is very similar to the previously uploaded one in most scenarios. In other words, the received local candidates from each monitor node do not often change dramatically as time moves. Therefore, we can record the statuses of candidate objects to make each monitor node only upload the objects that need to be updated. Such a way can alleviate the transmission cost mentioned in the second issue.

In this paper, we use a table, checking-time table ($CT$), to record the statuses of received data objects on the coordinator node $N_H$. The status of each received data object $u$ will be stored in one entry of $CT$ until the lifetime of $u$ is out. Note that the lifetime of a data object is equal to the length of $N_H$'s sliding window $SW_H$. The coordinator node thus only needs to update the status of data objects in $CT$ if necessary, and then calculates the final result. In general, most of the objects will not be in the final result. Hence, we can use a predictive way to determine the minimum checking time for the coordinator node. With $CT$ and the minimum checking time derivation, the server can only do the computation if the result will change. Such an idea comes from the conventional work~\citep{10.1007/978-3-319-11749-2_19} in centralized database systems. We thereby propose a new distributed version theorem for dynamically determining the minimum checking time to update the result set $PTD$ in distributed environments.
%\vspace{.3em}
\begin{thm}[\textbf{Minimum Checking Time}]
	\label{thm_mct}
	%Given an uncertain data object $u_i$ and a set of probabilistic top-$k$ dominating objects $PTD$, if $u_i\notin PTD\wedge u_i\in CS$, $dom(u_i)$ will be the dominant score of $u_i$, where $CS$ is the candidate set.
	Suppose that the notations are defined as above, the minimum checking time for the coordinator node represents the lower bound of expected time that the result set of probabilistic top-$k$ dominating objects $PTD$ will change, and it can be derived by
	\begin{equation}\label{eq:mct}
	mct(u)=\min(exp_{min},\lfloor\dfrac{dom_k-dom(u)}{mn} \rfloor+t_{cur}),
	\end{equation}
	where $exp_{min}$ is the nearest (smallest) expired time of an object in the set of $PTD$, $dom_k$ is the $k$-th highest dominant score of the objects in $PTD$, $t_{cur}$ is current time.
\end{thm}
\begin{proof}
	When $u\in CS$ and $u\notin PTD$, $u$ has a chance to be in $PTD$ if one of following two cases is satisfied:
	\begin{enumerate}
		\item some objects in $PTD$ are expired;
		\item $dom(u)\geq dom_k$.
	\end{enumerate}
	Our objective is to obtain the minimum time that $u$ can be in $PTD$. For Case 1, the process will search the minimum expired time of all objects in $PTD$ and it is depicted as $exp_{min}$. For Case 2, if object $u_k$ is the object with $k$-th highest dominant score in $PTD$ and object $u_{old}$ is going to be removed from $PTD$. Removing $u_{old}$ from $PTD$ will reduce the difference gap, $dom_k-dom(u)$, between $u$ and $u_k$. It means that some objects $u_{old}$ in $PTD$ result in $dom(u)\geq dom_k$. In general, there could be many old objects like $u_{old}$. The system needs to remove at least $\lfloor \dfrac{dom_k-dom(u)}{n}\rfloor$ old objects to remain the minimum number of necessary objects in $PTD$, and then $dom(u)\geq dom_k$ holds. Since each run of the computation can remove $m$ old objects like $u_{old}$, the minimum time period needs to be divided by $m$ and it will be $\lfloor \dfrac{dom_k-dom(u)}{mn}\rfloor$. After that, add the obtained minimum time period to the current time and thus get the predicted time, $\lfloor \dfrac{dom_k-dom(u)}{mn}\rfloor+t_{cur}$, that $dom(u)\geq dom_k$. Finally, the minimum checking time is updated by $\min(exp_{min},\lfloor\dfrac{dom_k-dom(u)}{mn} \rfloor+t_{cur})$ with respect to an object $u$, which is denoted as $mct(u)$.
\end{proof}
%\vspace{.3em}
While computing $mct(u)$ of each object $u$, if the obtained time is equal to $t_{cur}$, it means that $u$ has a chance to become the result in this run but the dominant score of $u$ is not large enough to make $u\in PTD$. Then, if $u\in SW_H$ during the next run of computation, the system will record each $u$ in the checking-time table $CT$ and set $mct(u) = t_{cur}+1$ for the next run of computation. In summary, with the proposed theorem, the checking-time table $CT$ acts as a priority cache table and helps the $N_H$ effectively reduce the frequency of computation for the third issue.

\subsection{The Process of PDTMUS}
In fact, the overall process of PDTMUS has been briefly described in Section~\ref{sec:def_ptdmus}. In this subsection, we show the whole process using the pseudo-code in Algorithm~\ref{alg:PDTMUS}.
%Note that the implementations of R-tree's management and the operations of threshold-based probabilistic $k$-skyband, are too complicated to present, so we will detailedly introduce the process of PDTMUS with processing state graphs alternatively for the ease of understanding.
The system executes Algorithm~\ref{alg:PDTMUS} recursively until no input data coming in. At Line~\ref{alg:PDTMUS:line1} of Algorithm~\ref{alg:PDTMUS}, each monitor node $N_j$ inserts the received data objects in $US_j$ to the local sliding widow $SW_j$ at time $t$ and removes the oldest data objects in $SW_j$ due to the size limitation of a local sliding widow. At Line~\ref{alg:PDTMUS:line2}, each $N_j$ pre-processes all the local objects in $SW_j$ for constructing the local R-tree $R_{j}$ as well as obtaining the information of dominant relations between MBRs using Definition~\ref{def_DMBR}. Each $N_j$ then computes the local candidate (local $k$-skyband) set $CS_j^t$ with Definition~\ref{def_dpo}, Definition~\ref{def_POD}, and Definition~\ref{def_TKS}. Note that $CS_j^t=TKS_j^t$ and $TKS_j^t$ is the local result of threshold-based top-$k$ dominating objects from the monitor node $N_j$. If $t=0$ holds at Line~\ref{alg:PDTMUS:line4}, it means that the whole precess is in the initial phase and each $N_j$ at Line~\ref{alg:PDTMUS:line5} will upload the whole candidate set $CS_j^t$ to the coordinator node $N_H$; otherwise, each $N_j$ only needs to upload the necessary update information to $N_H$ at Line~\ref{alg:PDTMUS:line7}.
\begin{algorithm2e}[t]
	\small
	\SetAlgoLined
	\KwIn{$N_H$, $N_j$ ($1\leq j\leq m$), time-stamp $t$, threshold $\delta$, result limit $k$, sliding windows $SW_H$, $SW_j$, checking-time table $CT$}
	\KwOut{the set of probabilistic top-$k$ dominating data objects, $PTD^t$}
	every $N_j$ inserts (remove) objects into (from) $SW_j$\label{alg:PDTMUS:line1}\;
	every $N_j$ pre-processes the objects in $SW_j$ to generate the local R-tree $R_j$\label{alg:PDTMUS:line2}\;
	every $N_j$ derives $dom(u)$ and $r$-$dom(u),\forall u\in CS_j^t$ by Definition~\ref{def_POD}, uses $\delta$ to preclude irrelevant objects in $CS_j^t$ by Definition~\ref{def_TKS}, and then obtains the local candidate set $CS_j^t=TKS_j^t$\label{alg:PDTMUS:line3}\;
	\uIf{$t==0$}{\label{alg:PDTMUS:line4}
		every $N_j$ uploads $CS_j^t$ to $N_H$\label{alg:PDTMUS:line5}\;
	}
	\Else{
		every $N_j$ uploads update information to $N_H$\label{alg:PDTMUS:line7}\;
	}
	$N_H$ computes the global candidate set $CS^t$ from $SW_H$\label{alg:PDTMUS:line9}\;
	$N_H$ broadcasts $CS^t$ to every $N_j$\label{alg:PDTMUS:line10}\;
	%$N_H$ computes the exact candidate set $exactCS$\label{alg:PDTMUS:line11}\;
	every $N_j$ derives $dom(u)$ and $r$-$dom(u),\forall u\in CS^t$ by Definition~\ref{def_POD} and uses $\delta$ to preclude irrelevant objects in $CS^t$ by Definition~\ref{def_TKS}\label{alg:PDTMUS:line11}\;
	every $N_j$ uploads the updated $CS^t$ to $N_H$ \label{alg:PDTMUS:line12}\;
	$N_H$ sums up the received scores of objects being dominated and update $CS^t$\label{alg:PDTMUS:line13}\;
	$N_H$ finds $PTD^t$ from $CS^t$\label{alg:PDTMUS:line14}\;
	$N_H$ uses~\eqref{eq:mct} to update the minimum checking time of each object in $CS^t-PTD^t$ and saves the information in checking time table $CT^t$\label{alg:PDTMUS:line15}\;
	$N_H$ broadcast $CT^t$ to every $N_j$\label{alg:PDTMUS:line16}\;
	\Return $PTD^t$\label{alg:PDTMUS:end}\;
	\caption{The main process of PDTMUS}
	\label{alg:PDTMUS}
\end{algorithm2e}

After $N_H$ receives each local candidate set from each $N_j$ in $SW_H$, $N_H$ derives the global candidate set $CS^t$ in the same way (using Definition~\ref{def_POD} and Definition~\ref{def_PKOD}) at Line~\ref{alg:PDTMUS:line9}.
The coordinator node $N_H$ then broadcasts the global candidate set to every $N_j$ at Line~\ref{alg:PDTMUS:line10} and asks $N_j$ for helping the local computation. Each $N_j$ derives the dominant and dominated scores, $dom(u)$ and $r$-$dom(u)$, of all the objects in $CS_j^t$ and uploads the updated $CS_j^t$ to $N_H$ using Definition~\ref{def_TKS} at Lines~\ref{alg:PDTMUS:line11} and~\ref{alg:PDTMUS:line12}. From Lines~\ref{alg:PDTMUS:line13} to \ref{alg:PDTMUS:line15}, $N_H$ uses the received information of dominated scores to update $CS^t$, finds the final global result $PTD^t$ for time $t$, and then updates the minimum checking times of the objects that may be the answer at time $t+exp_{min}$. $N_H$ broadcasts the information of checking time in $CT^t$ to every $N_j$ at Line~\ref{alg:PDTMUS:line16} using Theorem~\ref{thm_mct}. With the checking time table $CT^t$, each $N_j$ can determine the appropriate time of the next round of update/derivation and thus effectively reduce the frequency of computation. Such a way can avoid a lot of unnecessary computation. In the last, the system returns $PTD^t$ as the final result to the user.

\section{Complexity Analysis}
\label{sec:analysis}
After introducing the proposed process of PTDMUS, we analyze and discuss its time complexity, space complexity, and transmission cost in both the average case and the worst case, respectively.

\subsection{Time Complexity}
In the first run (time slot $t=0$) of the PTDMUS process, mentioned in the previous section, each monitor node $N_j$ takes time on constructing a local R-tree, $R_j$, with all the data objects in $SW_j$ at the initial step, deriving $dom(u)$ and $r$-$dom(u)$ of each $u$ in $R_j$ at the second step, and extracting the threshold-based top-$k$ dominating objects into $TKS_j$ at the last step. Hence, the time complexity of PTDMUS on a monitor node $N_j$ can be expressed as
\begin{align}\label{eq:monitor:time:previous}
\text{T}_{\text{average}}(N_j,t=0)=&\text{T}_{\text{construction}}(R_j)+\text{T}_{\text{update}}(R_j)\nonumber\\
+&\text{T}_{\text{extract}}(TKS_j).
\end{align}
In PTDMUS, the time complexity is related to maintaining and searching the R-trees. According to~\citep{ALBORZI20076}~\citep{Arge:2008:PRP:1328911.1328920}, the time for constructing a $d$-dimensional R-tree is $O(\dfrac{|U|}{B}\log_{R_{degree}/B}\dfrac{|U|}{B})$ where $B$ is the block (or page) size of data on the disk (or memory), $R_{degree}$ is the degree fanout of R-tree. In this work, we deal with the uncertain data objects in a object-oriented model ($B=1$), so the time for local R-tree's construction will be
\begin{align}\label{eq:monitor:time:local_construction}
\text{T}_{\text{construction}}(R_j)=|SW_j|\log_{R_{degree}}|SW_j|.
\end{align}

In the considered environment, we assume
that the data points are uniformly and independently distributed in the domain space $[0, 2000]^d$. To make it simple to analysis, we normalize the space into $[0,1]^d$.
According to~\citep{Theodoridis:1996:MPR:237661.237705}, $R_j$'s height $h_j$ and the number of nodes $N_L$ at level $L$ (let the leaf level be $0$) will be approximately $h_j = 1 + \lceil\log_{R_{degree}} (|SW_j|/R_{degree})\rceil$ and $N_L=|SW_j|/(R_{degree})^{L+1}$, respectively. Besides, the extent $\theta_{L}$ (i.e., length of any 1D projection) of a node at the $L$-th level can be estimated by $\theta_{L}=(1/N_L)^{1/d}$ and some nodes in the $L$-th level may be partially dominated by $u$. Fig.~\ref{fig:upper_bound:dom} shows that the gray region $I_2$ corresponds to the maximal region, covering nodes (at level $L$) that are partially dominated by $u^{\min}$.
Then, the average number of required node accesses in the R-tree for computing the dominant score $dom(u)$ of object $u$ will be~\citep{Yiu:2007:EPT:1325851.1325908}
%\vspace{-8pt}
\begin{align*}
\text{T}_{dom(u)}(SW_j)=&\sum_{L=0}^{h_j-1} N_L n^2\times\\
 &[(1-v_{u^{\min}}+\theta_{L})^d-(1-v_{u^{\min}}-\theta_{L})^d],
\end{align*}
where $v_{u^{\min}}$ is the value of $u^{\min}$ after the 1D projection and $n$ is the number of instances in an object. Hence, the time complexity of dominance update on the monitor node can be expressed as
\begin{align}\label{eq:monitor:time:dom}
\text{T}_{dom}(N_j,t=0)=|SW_j|\times\text{T}_{dom(u)}(SW_j).
\end{align}
To obtain the local threshold-based probabilistic $k$-skyband, the monitor node $N_j$ also needs to traverse the $R_j$ to derive the $r$-$dom(u)$ of each object $u$ in the $SW_j$. Fig.~\ref{fig:upper_bound:rdom} shows that the gray region $I_2'$ corresponds to the maximal region, covering nodes (at level $L$) that partially dominate $u^{\max}$. The average number of required node accesses in the R-tree for computing the $r$-$dom(u)$ of object $u$ will be
\begin{align*}
\text{T}_{r\text{-}dom(u)}(SW_j)=&\sum_{L=0}^{h_j-1} N_L n^2\times\\
 &[(1-v_{u^{\max}}-\theta_{L})^d-(1-v_{u^{\max}}-2\theta_{L})^d],
\end{align*}
where $v_{u^{\max}}$ is the value of $u^{\max}$ after the 1D projection. Hence, the time complexity of $k$-skyband update on a monitor node can be derived by
\begin{align}\label{eq:monitor:time:rdom}
\text{T}_{r\text{-}dom}(N_j,t=0)=|SW_j|\times\text{T}_{r\text{-}dom(u)}(SW_j).
\end{align}

\begin{figure}[ht]
	\vspace{-15pt}
	\centering
	\subfigure[Computing $dom(u)$]{
		\label{fig:upper_bound:dom} %% label for 1st subfigure
		\includegraphics[width=0.23 \textwidth]{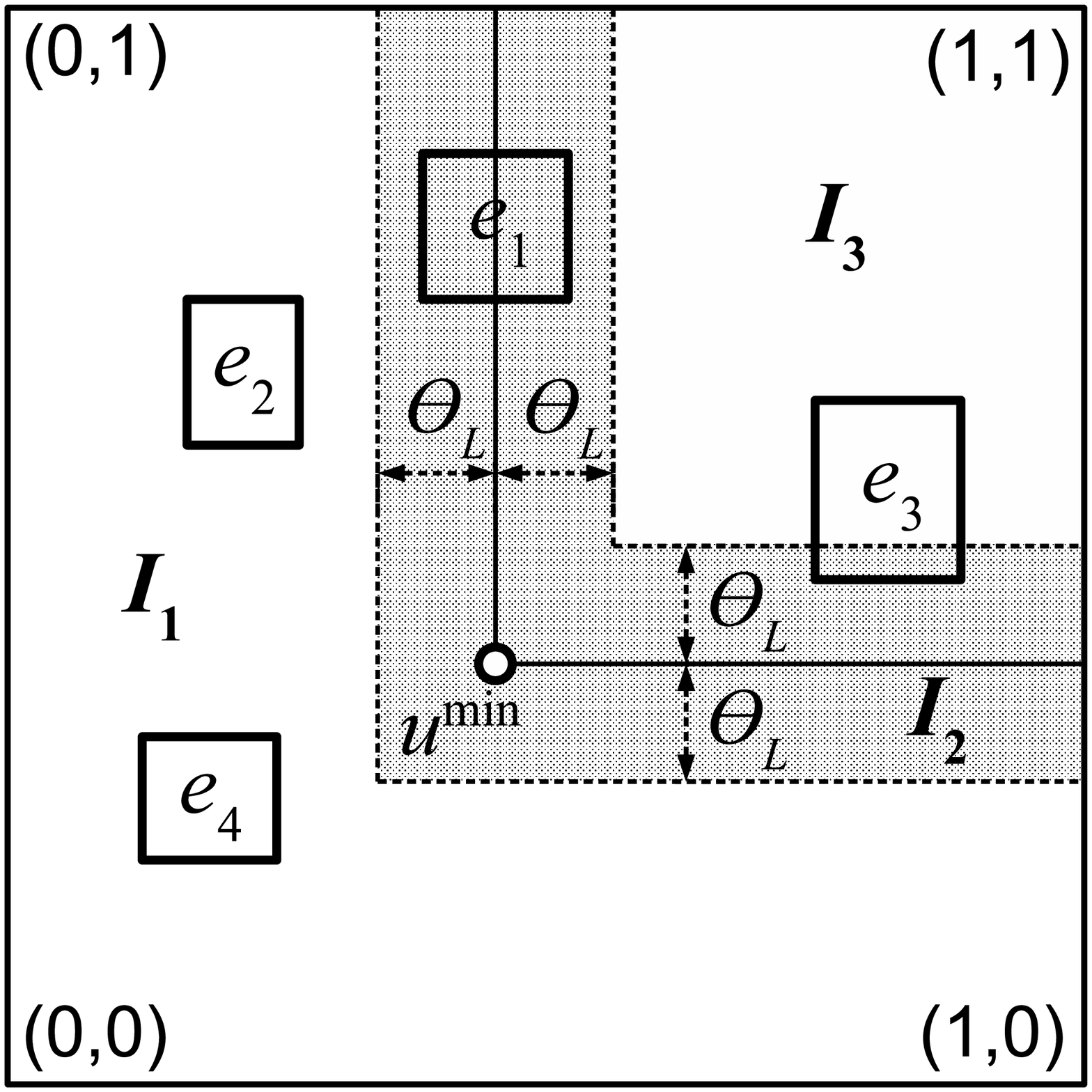}}%\hspace{1.5em}
	\subfigure[Computing $r$-$dom(u)$]{
		\label{fig:upper_bound:rdom} %% label for 2nd subfigure
		\includegraphics[width=0.23 \textwidth]{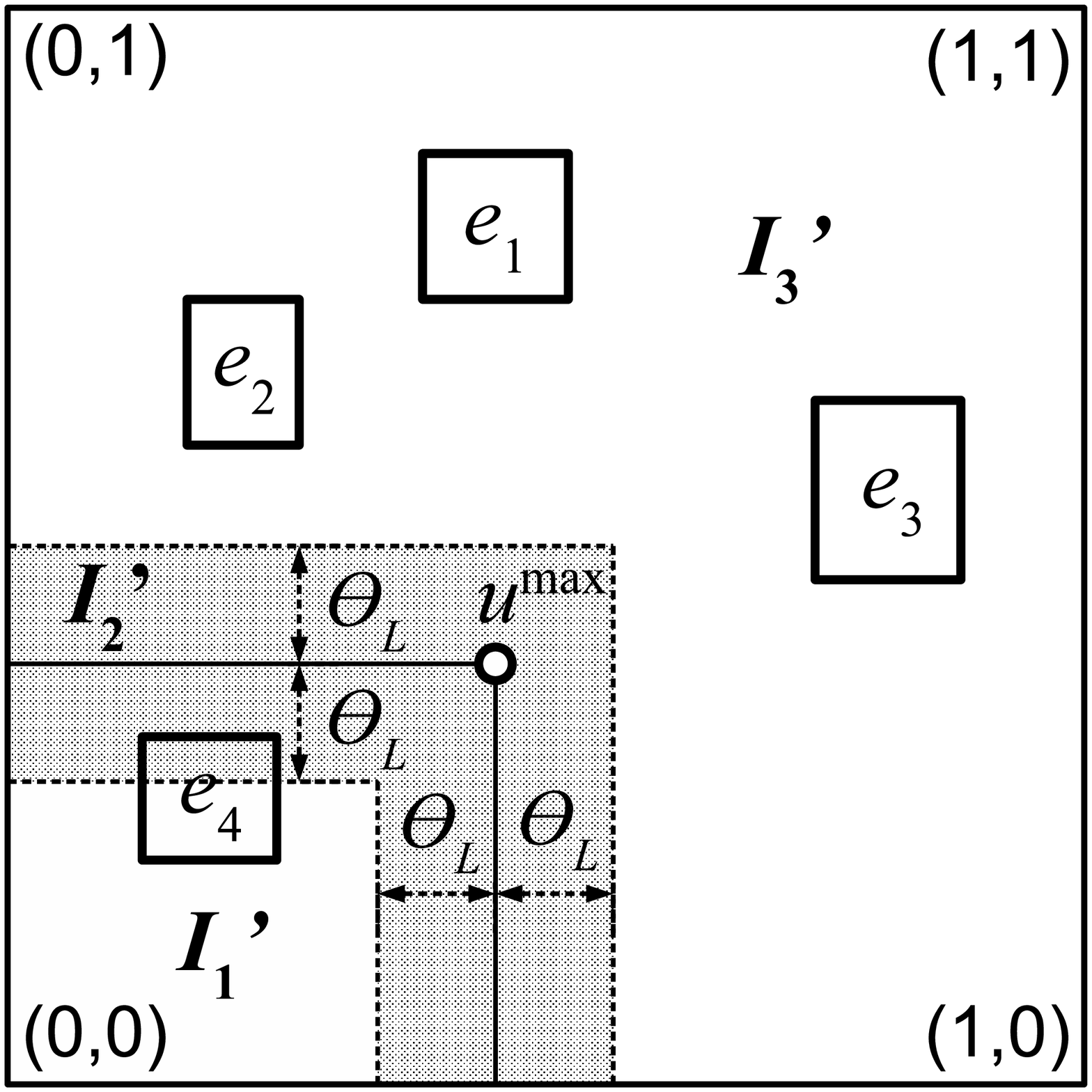}}%\hspace{.5in}
	\caption{The approximated upper bounds of computational costs on \subref{fig:upper_bound:dom} Computing $dom(u)$ , \subref{fig:upper_bound:rdom} Computing $r$-$dom(u)$.
	}
	\label{fig:upper_bound} %% label for entire figure
	%\vspace{-5pt}
\end{figure}

That is, with~\eqref{eq:monitor:time:dom} and~\eqref{eq:monitor:time:rdom}, the time complexity of the second step on $N_j$ will be
{\small
\begin{align}\label{eq:monitor:time:update}
\text{T}_{\text{update}}(R_j)=&\text{T}_{dom}(N_j,t=0)+\text{T}_{r\text{-}dom}(N_j,t=0)\nonumber\\
=&|SW_j|\times(\text{T}_{dom(u)}(SW_j)+\text{T}_{r\text{-}dom(u)}(SW_j)),
\end{align}}%
where $\forall u\in SW_j$.
In the last step, $N_j$ will copy all the objects in $SW_j$ to a temporary candidate list $CS_j'$, sort the objects in decreasing order by $dom(u)$ using merge-sort where $\forall u\in CS_j'$, and then use the threshold $\delta$ to extract the local probabilistic $k$-skyband. Therefore, the time complexity of $\text{T}_{\text{extract}}(TKS_j)$ can be denoted as
\begin{align}
\text{T}_{\text{extract}}(TKS_j)=2\times|SW_j|+|SW_j|\log_{2}|SW_j|.\nonumber
\end{align}
In summary, with~\eqref{eq:monitor:time:local_construction} to~\eqref{eq:monitor:time:update}, we can express~\eqref{eq:monitor:time:previous} as

{\small\vspace{-10pt}
\begin{align}\label{eq:monitor:time:final}
\text{T}_{\text{average}}(N_j,t=0)=&|SW_j|\times(\log_{R_{degree}} |SW_j|+\text{T}_{dom(u)}(SW_j)\nonumber\\
+&\text{T}_{r\text{-}dom(u)}(SW_j)+2+\log_2 |SW_j|),
\end{align}}%
where $\forall u\in SW_j$.

Note that PTDMUS needs to monitor the result of the top-$k$ dominating query continuously in a monitoring time period $\varDelta t$ and $\varDelta t$ is set by

{\small\vspace{-10pt}
\begin{align}\label{eq:monitor_time}
\varDelta t=&\begin{cases}
\max_{1\leq j\leq m}\{\dfrac{|U|}{m}-|SW_j|\}, & \text{if $\dfrac{|U|}{m}-|SW_j|>0$}.\\
1, & \text{otherwise}.
\end{cases}
\end{align}}%
After the first run (time slot), the coordinator node $N_H$ will broadcast the global candidate set $CS^t$ at time $t$ with the minimum checking time, $exp_{min}$, to each monitor node $N_j$. Each $N_j$ can use the received information to reduce the computational overhead during the next computation of local result when $t>0$. In the following time slots, $N_j$ uses the candidate set of previous run, $CS^{t-exp_{min}}$, to construct the global R-tree, $R_H^t$, for dominance checks instead of using $R_j$. In practice, we use two temporary lists, $DO_j^t$ and $NO_j^t$, to help the update of candidate set during the time period $(t-exp_{min},t]$. $DO_j^t$ is used to record the objects that are going to be deleted where $DO_j^t=\{SW_j^{t-exp_{min}}[0],SW_j^{t-exp_{min}}[1],\dots, SW_j^{t-exp_{min}}[exp_{min}-1]\}$. $NO_j^t$ is used to stored the new input objects that are going to be added where $NO_j^t=\{SW_j^t[|SW_j|-1],SW_j^t[|SW_j|-2],\dots,SW_j^t[|SW_j|-exp_{min}]\}$. Thus, the exact data set $UO_j^t$ that needs to be processed at time $t$ becomes $CS^{t-exp_{min}}\cup NO_j^t-DO_j^t$ and $N_j$ uses $UO_j^t$ to construct the new local $R_j^t$ for computing $TKS_j^t$. Using $UO_j^t$ to substitute $SW_j$ with~\eqref{eq:monitor:time:local_construction} to~\eqref{eq:monitor:time:final}, the time complexity of a derivation run on $N_j$ at time $t$ can be obtained by

{\small\vspace{-10pt}
\begin{align}\label{eq:monitor:time:t_larger_0:previous}
\text{T}_{\text{average}}(N_j,t>0)=&|UO_j^t|\times(\log_{R_{degree}} |UO_j^t|+\text{T}_{dom(u)}(UO_j^t)\nonumber\\
+&\text{T}_{r\text{-}dom(u)}(UO_j^t)+2+\log_2 |UO_j^t|),
\end{align}}%
where $\forall u\in UO_j^t$.
In summary, the average complexity during the time $\varDelta t$ will be

{\small\vspace{-10pt}
\begin{align}\label{eq:monitor:average:time}
\text{T}_{\text{average}}(N_j)=&\dfrac{1}{\varDelta t}(\text{T}_{\text{average}}(N_j,t=0)\nonumber\\
+&\sum_{h=1}^{\lfloor\varDelta t/\overline{exp_{min}}\rfloor}\text{T}_{\text{average}}(N_j,t=h\times\overline{exp_{min}})).
\end{align}}%

In fact, the derived costs $\text{T}_{\text{average}}(N_j,t=0)$ in PTDMUS and PTDSky methods are similar since both of them use monitor nodes to derive the local $k$-skybands. From~\eqref{eq:monitor:average:time}, we can know that the computation time is significantly influenced by the computation cost of each run when $t>0$. PTDMUS uses the minimum checking time $exp_{min}$ to reduce the frequency $f$ of derivations (or dominance checks) where $f=1+\lfloor\varDelta t/\overline{exp_{min}}\rfloor$. If $\overline{exp_{min}}=1$, each monitor node in PTDMUS and PTDSky will have similar computation time $\text{T}_{\text{average}}(N_j)$. The worst case only occurs when $\overline{exp_{min}}=1$ and $N_j$ always receives the global candidate set $CS^t=SW_H$. In such a scenario, the set $UO_j^t$' needed to be process at each time slot $t$ will be $UO_j^t$'$=SW_H^{t-exp_{min}}\cup NO_j^t-DO_j^t$ and $|UO_j^t$'$|$ will become very large. To obtain the upper bound of the time complexity, $\text{T}_{\text{worst}}(N_j)$, on a monitor $N_j$, we can use $\overline{exp_{min}}=1$ and substitute $UO_j^t$' for $UO_j^t$ in~\eqref{eq:monitor:time:t_larger_0:previous} and~\eqref{eq:monitor:average:time}.

After analyzing the time complexity on a monitor node, time complexity on the coordinator node, $\text{T}_{\text{average}}(N_H)$, also needs to be discussed. However, $\text{T}_{\text{average}}(N_H)$ depends on the size of global candidate set $CS$, so we will discuss $\text{T}_{\text{average}}(N_H)$ after analyzing the space complexity on $N_H$ in the next subsection.

\subsection{Space Complexity}
In the considered parallel computing model, the size of global candidate set $|CS|$ in the coordinator node $N_H$ depends on the size of received local threshold-based probabilistic $k$-skyband $|TKS_j|$ from each monitor node $N_j$ and the number of monitor nodes $m$.
Suppose that $\text{PKsky}(u)$ is an indicator function defined as
\begin{align*}
\text{PKsky}(u)=&\begin{cases}
1, & \text{if $dom(u)\geq 1\wedge r\text{-}dom(u)<\delta$}.\\
0, & \text{otherwise}.
\end{cases}
\end{align*}
In most application scenarios of big data, the size of local result, $|TKS_j|$, is usually larger than $k$. Thus, the average size of global candidate set $\overline{|CS|}=\text{SP}_{\text{average}}(CS)$ will be
\begin{align}\label{eq:space:average}
\text{SP}_{\text{average}}(CS)=&\sum_{j=1}^m \sum_{l=1}^{|SW_j|} \text{PKsky}(u)\nonumber\\
=&\sum_{j=1}^m |TKS_j|=m\times \overline{|TKS|},
\end{align}
where $\overline{|TKS|}$ is the average size of the received local threshold-based probabilistic $k$-skybands from the monitor nodes. Note that both $|CS|$ and $|TKS_j|$ are usually much larger than $k$ in most big data applications. In general, $|TKS_j|$ is much smaller than $|SW_j|$ due to the dominance and object pruning by the threshold.

Consider the worst case, the space complexity of candidate set $CS$ in the monitor node $N_H$ can be denoted as
\begin{align}\label{eq:space:worst}
\text{SP}_{\text{worst}}(CS)=&\sum_{j=1}^m |SW_j|=m\times \overline{|SW|}=|SW_H|,
\end{align}
where $\overline{|SW|}$ is the average size of the sliding windows in monitor nodes. The worst case only happens when all the uncertain data objects are anti-correlated in all dimensions. It means that the condition $\forall u\in U, dom(u)=0\wedge r$-$dom(u)=0$ holds and makes all the data objects in monitor nodes to be uploaded to the coordinator node. Thus, the space complexity of the worst case in the monitor node $N_H$ is $O(|SW_H|)$. However, such a case is almost impossible to occur in big data environments.

After discussing the average and the worst space complexities on the coordinator node $N_H$ respectively, we can start discussing the time complexity of the computation on $N_H$. In PTDMUS, instead of computing the global $k$-skyband, $N_H$ just uses merge-sort to sort the received data objects from the monitor nodes by $dom(u)$ in a decreasing order, derives the expected checking time of $u$, and finds the minimum checking time $exp_{min}$ at each run (time $t$), where $\forall u\in CS^t$. Hence, the average time complexity for one run on $N_H$, $\text{T}_{\text{average}}(N_H,t\geq 0)$, can be formulated as
\begin{align*}%\label{eq:coordinator:average:time:onerun}
\text{T}_{\text{average}}(N_H,t\geq 0)=&|CS^t|\log_2 |CS^t|+|CS^t|.
\end{align*}
Due to the usage of the minimum checking time, the expected average time complexity can be derived by
\begin{align*}%\label{eq:coordinator:average:time}
\text{T}_{\text{average}}(N_H)=&\dfrac{1}{\varDelta t}\sum_{h=0}^{\lfloor\varDelta t/\overline{exp_{min}}\rfloor}\text{T}_{\text{average}}(N_H,t=h\times\overline{exp_{min}})),\\
=&\dfrac{f}{\varDelta t}\times(\overline{|CS|}\log_2\overline{|CS|}+\overline{|CS|}),
\end{align*}
where $\overline{|CS|}=\text{SP}_{\text{average}}(CS)$ in~\eqref{eq:space:average} and $f=1+\lfloor\varDelta t/\overline{exp_{min}}\rfloor$.
Additionally, the worst case occurs when $\overline{exp_{min}}=1$ (or $f=\varDelta t$) and \eqref{eq:space:worst} holds. Then the worst time complexity can be obtained by
\begin{align*}%\label{eq:coordinator:worst:time}
\text{T}_{\text{worst}}(N_H)=&|SW_H|\log_2|SW_H|+|SW_H|.
\end{align*}

\subsection{Transmission Cost}
In general, the transmission cost depends on the sizes of local probabilistic $k$-skybands and the global candidate set. According to the process of PTDMUS in Algorithm~\ref{alg:PDTMUS}, the average transmission cost of a monitor node can be expressed as
\begin{align}
Cost_{\text{average}}=&\dfrac{1}{\varDelta t}\times(Cost_{\text{initial}}+Cost_{\text{update}}),\label{eq:avg_cost:1}\\
Cost_{\text{initial}}=&\overline{|TKS^{t=0}|}+2\times\overline{|CS^{t=0}|}+\overline{|CT^{t=0}|},\text{ and}\label{eq:avg_cost:2}\\
Cost_{\text{update}}=&\sum_{h=1}^{\lfloor\varDelta t/\overline{exp_{min}}\rfloor}(\overline{|Info_{\text{update}}|}\nonumber\\
+&2\times\overline{|CS^{t=h\times\overline{exp_{min}}}|}+\overline{|CT^{t=h\times\overline{exp_{min}}}|}),\label{eq:avg_cost:3}
\end{align}%
where $TKS^{t=0}$, $CS^{t=0}$, and $CT^{t=0}$ are respectively the local threshold-based $k$-skyband, candidate set, and checking time table at the initial step (the fist time slot), as well as $Info_{\text{update}}$ is the minimum set of candidate objects needed to be updated at the $h\times\overline{exp_{min}}$ time slot. Note that $Info_{\text{update}}$ is expressed as

{\small\vspace{-10pt}
\begin{align}\label{eq:avg_cost:4}
Info_{\text{update}}=&\bigcup{}\{u\in TKS^{t=h\times\overline{exp_{min}}}|u'\in TKS^{t=(h-1)\times\overline{exp_{min}}}\wedge \nonumber\\ &u.ID=u'.ID\wedge(dom(u)!=dom(u')\vee \nonumber\\ &r\text{-}dom(u)!=r\text{-}dom(u'))\}.
\end{align}}%
In general, $|Info_{\text{update}}|$ is much smaller than $|TKS|$ and $|CS|$. With~\eqref{eq:avg_cost:3}, PTDMUS only needs to upload the update information $f=1+\lfloor\varDelta t/\overline{exp_{min}}\rfloor$ times during the monitor time $\varDelta t$. By contrast, PTDSky needs to upload information at every time slot of $\varDelta t$. In summary, Equations~\eqref{eq:avg_cost:1} to~\eqref{eq:avg_cost:4} are used to measure the average transmission cost of PTDMUS in the simulations.

The worst case of update cost only occurs when each input data object in the consequence time slots always becomes the top-$1$ dominating object. In this case, $Info_{\text{update}}$ will become $TKS^t$ and the monitor node always needs to upload $TKS^t$ at every time slot. In addition, the worst transmission cost on the information exchange between $N_H$ and $N_j$ occurs when $|CS^t|=|SW_H|=|CT^t|$. Hence, the worst transmission cost (or network load) of a monitor node will be
\begin{align*}
Cost_{\text{worst}}=&\dfrac{1}{\varDelta t}\times(\sum_{t=0}^{\varDelta t-1}(|TKS^t|+3\times|SW_H|)).
\end{align*}

\section{Simulation Results}
\label{sec:simulation}
The simulation including all compared approaches are implemented in JAVA with Spark using Eclipse IDE and the developed program is platform-independent. The simulation program is executed on a Windows 10 server with an Intel(R) Core(TM) i7-3770 CPU @ 3.40GHz - 3.80GHz and 8GB $\times$ 2 memory. In this simulation, we use synthetic data and the number of uncertain data objects is 10,000. We perform three different approaches for comparisons:
\begin{itemize}
	\item PTDMUS performs with R-trees, threshold-based probabilistic $k$-skyband in the monitor nodes, and PTDMUS performs with R-trees and the minimum checking time in the coordinator node;
	\item PTDSky executes with R-trees and threshold-based probabilistic $k$-skyband in both monitor and coordinator nodes~\citep{Mouratidis:2006:CMT:1142473.1142544};
	\item PTDBF only runs with R-trees in a centralized way without any parallelism.
\end{itemize}
%The performance of the above compared approaches are measured in terms of the computation time, transmission cost, precision, and recall, while considering the effects of threshold $\delta$, data dimensionality, the number of monitor nodes, the size of sliding window, the number of instances, the margin of uncertainty, the value of $k$, and the degree of R-tree.

Since PTDBF is performed in a centralized server with the global information including all input data streams, PTDBF can always has the correct result of a top-$k$ dominating query. Hence, PTDBF is treated as the baseline method in the simulation.
The performance of the above compared approaches is measured in terms of the \emph{computation time}, \emph{transmission cost}, \emph{precision}, and \emph{recall}, while considering the effects of threshold $\delta$, data dimensionality, the number of monitor nodes, the size of sliding window, the value of $k$, and the margin of uncertainty. In the previous section, both computation time and transmission cost have been detailedly analyzed in the average case and the worst case. The correctness/reliability of the proposed method is also important and thereby we validate the above methods in the simulation in terms of precision and recall.
Suppose that $PTD^t_{\text{baseline}}$ is the result set of top-$k$ dominating objects obtained from PTDBF at time $t$ and $PTD^t_{\text{compared}}$ is the one obtained from PTDMUS or PTDSky at time $t$, the precision and recall can be obtained by

{\small\vspace{-10pt}
\begin{align*}
\text{Precision}=&\dfrac{1}{\varDelta t}\times(\sum_{t=0}^{\varDelta t-1}(\dfrac{|PTD^t_{\text{baseline}}\cap PTD^t_{\text{compared}})|}{|PTD^t_{\text{compared}}|})\times 100\%,\\
\text{Recall}=&\dfrac{1}{\varDelta t}\times(\sum_{t=0}^{\varDelta t-1}(\dfrac{|PTD^t_{\text{baseline}}\cap PTD^t_{\text{compared}})|}{|PTD^t_{\text{baseline}}|})\times 100\%.
\end{align*}}%

We perform the simulations in 20 different scenarios and each scenario is executed $\varDelta t$ runs (time slots) to get the average results and $\varDelta t$ is set by~\eqref{eq:monitor_time}. The detailed setting of parameters is presented in Table~\ref{simulation_parameters}.
\begin{table}[ht]
	\caption{Simulation Parameters}
	\label{simulation_parameters}
	\centering
	%\scriptsize
	\begin{tabular}{p{3.7cm}cc}
		\hline
		\textbf{Parameter} & \textbf{Default Value} & \textbf{Range (type)}\\
		\hline
		Number of data objects, $|U|$ & $10000$ & - \\
		Number of instances, $n$ & $5$ & - \\
		Dimension, $d$ & $9$ & $3,5,7,9$\\		
		Space of an attribute & $[0, 2000]$ & - \\
		Number of monitor nodes, $m$ & $10$ & $4,6,8,10$\\
		Size of a local sliding window, $|SW_j|$  & $960$ & $240,480,720,960$ \\
		Size of the global sliding window, $|SW_H|$ & $9600$ & $m\times |SW_j|$\\
		Degree of R-tree, $R_{degree}$ & $6$ & - \\
		Threshold, $\delta$ & $30$ & $10, 20, 30, 40, 50$\\
		Margin of Uncertainty, $M$ & $160$ & $80, 160, 240, 320$ \\
		Distribution & Uniform & -\\
		$k$ & $100$ & $50,100,150,200$\\
		\hline
	\end{tabular}
\end{table}

\begin{figure*}[ht]
	\centering
	\subfigure[Computation Time]{
		\label{fig:threshold:time} %% label for 1st subfigure
		\includegraphics[width=0.24 \textwidth]{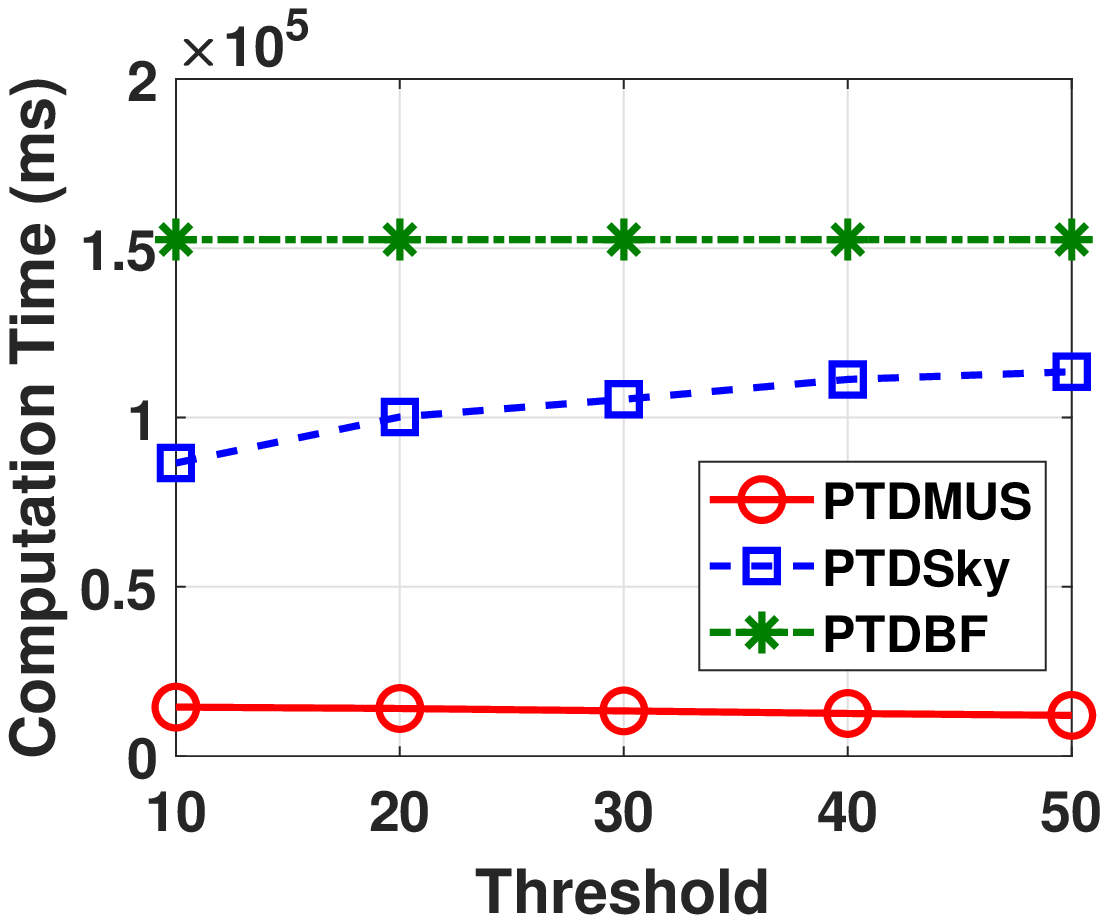}}%\hspace{.05in}
	\subfigure[Transmission Cost]{
		\label{fig:threshold:transmission} %% label for 2nd subfigure
		\includegraphics[width=0.24 \textwidth]{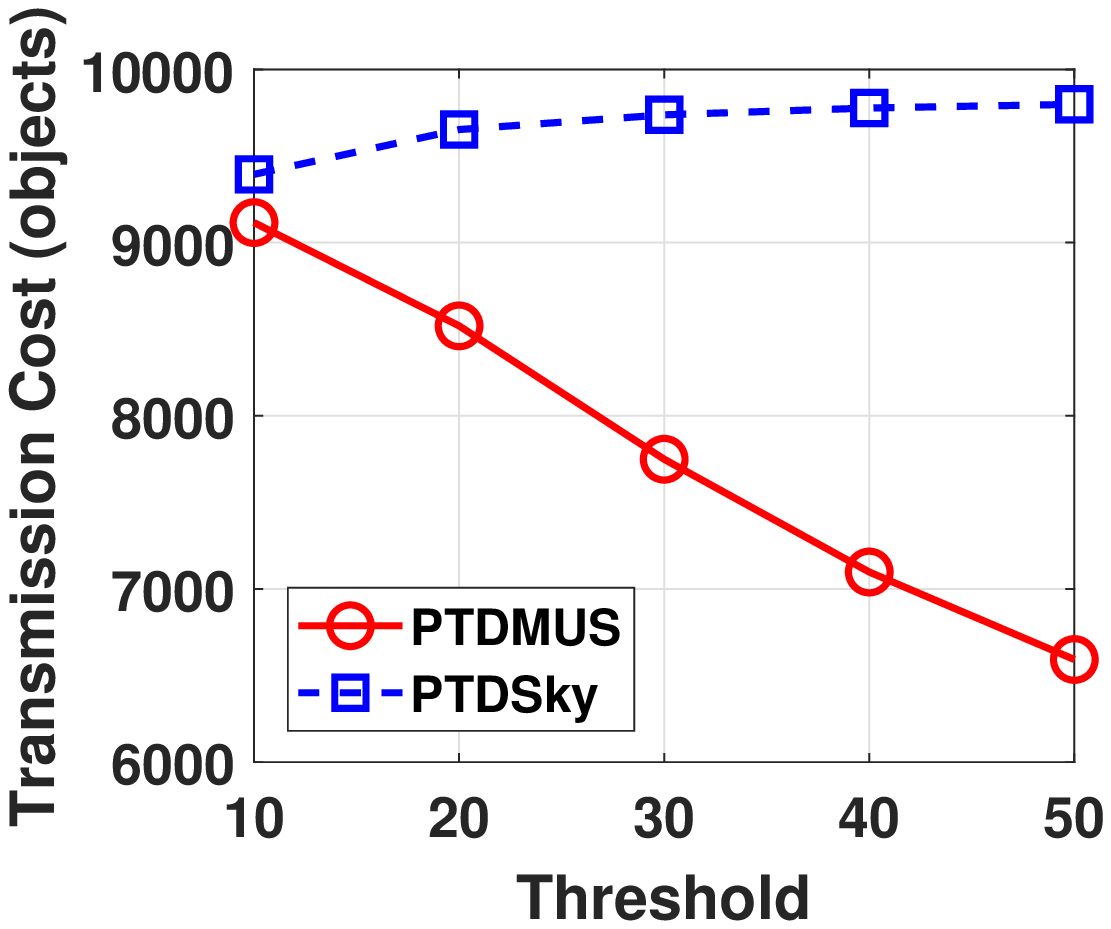}}%\hspace{.5in}
	\subfigure[Precision]{
		\label{fig:threshold:precision} %% label for 1st subfigure
		\includegraphics[width=0.24 \textwidth]{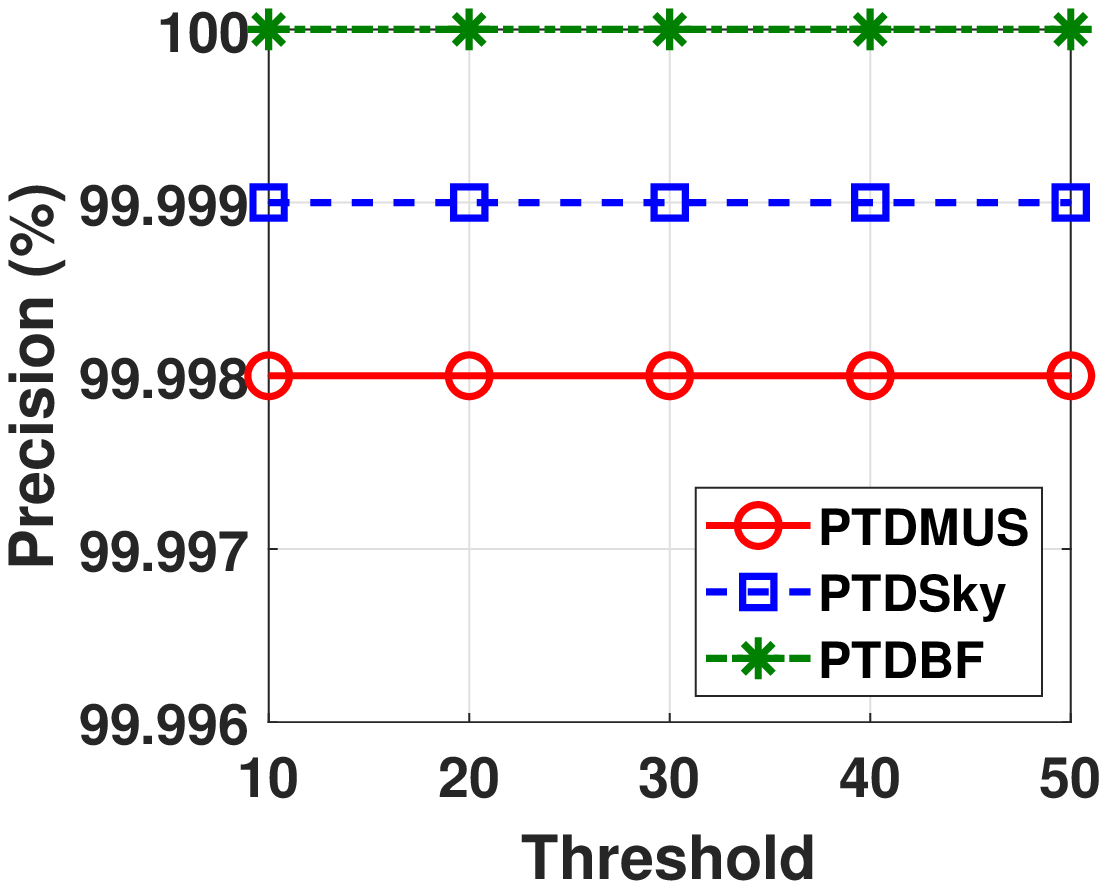}}%\hspace{.2in}
	\subfigure[Recall]{
		\label{fig:threshold:recall} %% label for 2nd subfigure
		\includegraphics[width=0.24 \textwidth]{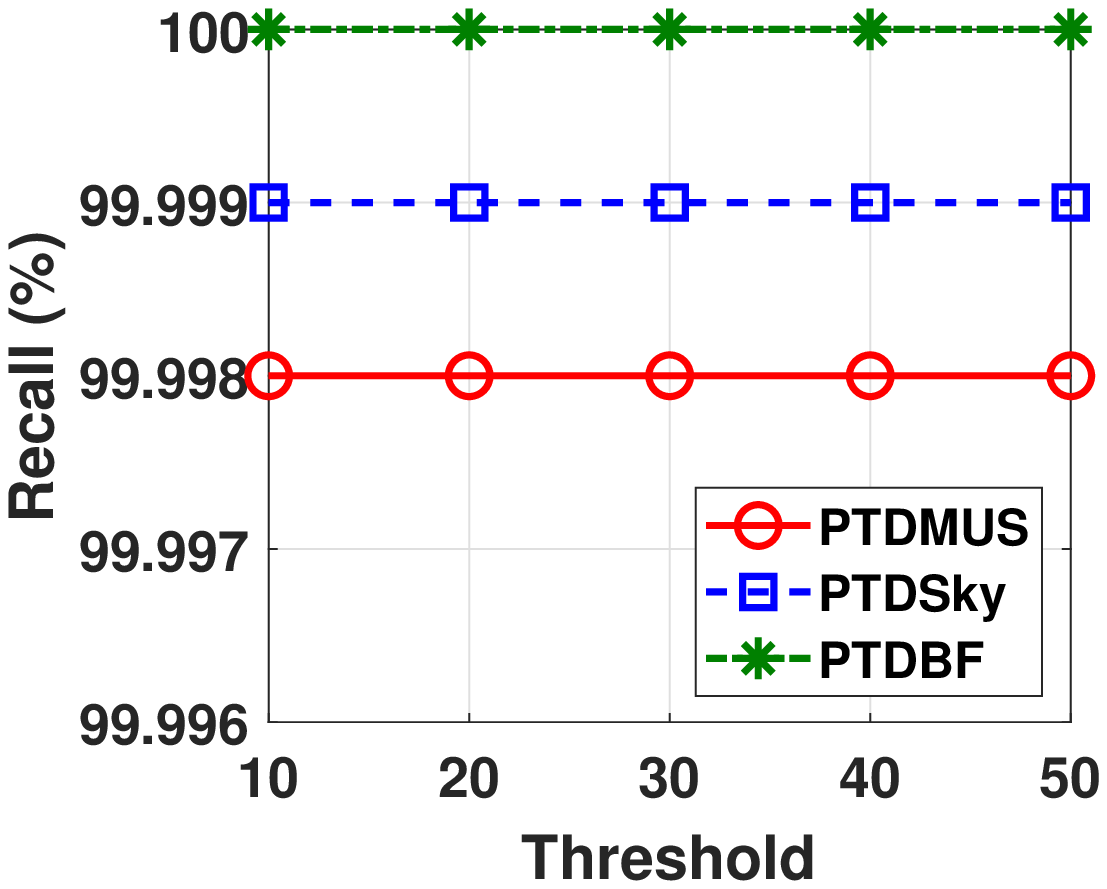}}%\hspace{.5in}
	\caption{Effect of threshold $\delta$ on \subref{fig:threshold:time} Computation Time, \subref{fig:threshold:transmission} Transmission Cost, \subref{fig:threshold:precision} Precision and \subref{fig:threshold:recall} Recall.
	}
	\label{fig:threshold} %% label for entire figure
	\vspace{-10pt}
\end{figure*}
\begin{figure*}[t]
	\centering
	\subfigure[Computation Time]{
		\label{fig:dimension:time} %% label for 1st subfigure
		\includegraphics[width=0.24 \textwidth]{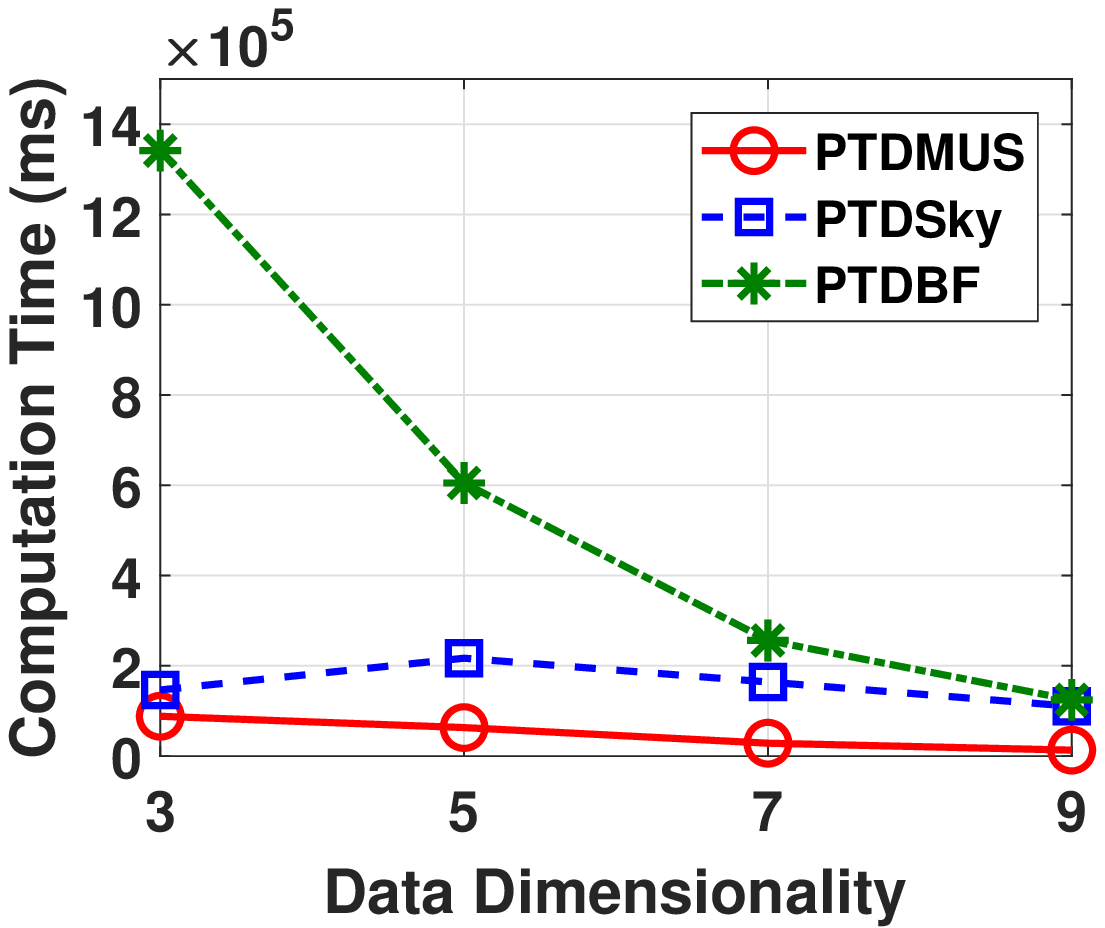}}%\hspace{.05in}
	\subfigure[Transmission Cost]{
		\label{fig:dimension:transmission} %% label for 2nd subfigure
		\includegraphics[width=0.24 \textwidth]{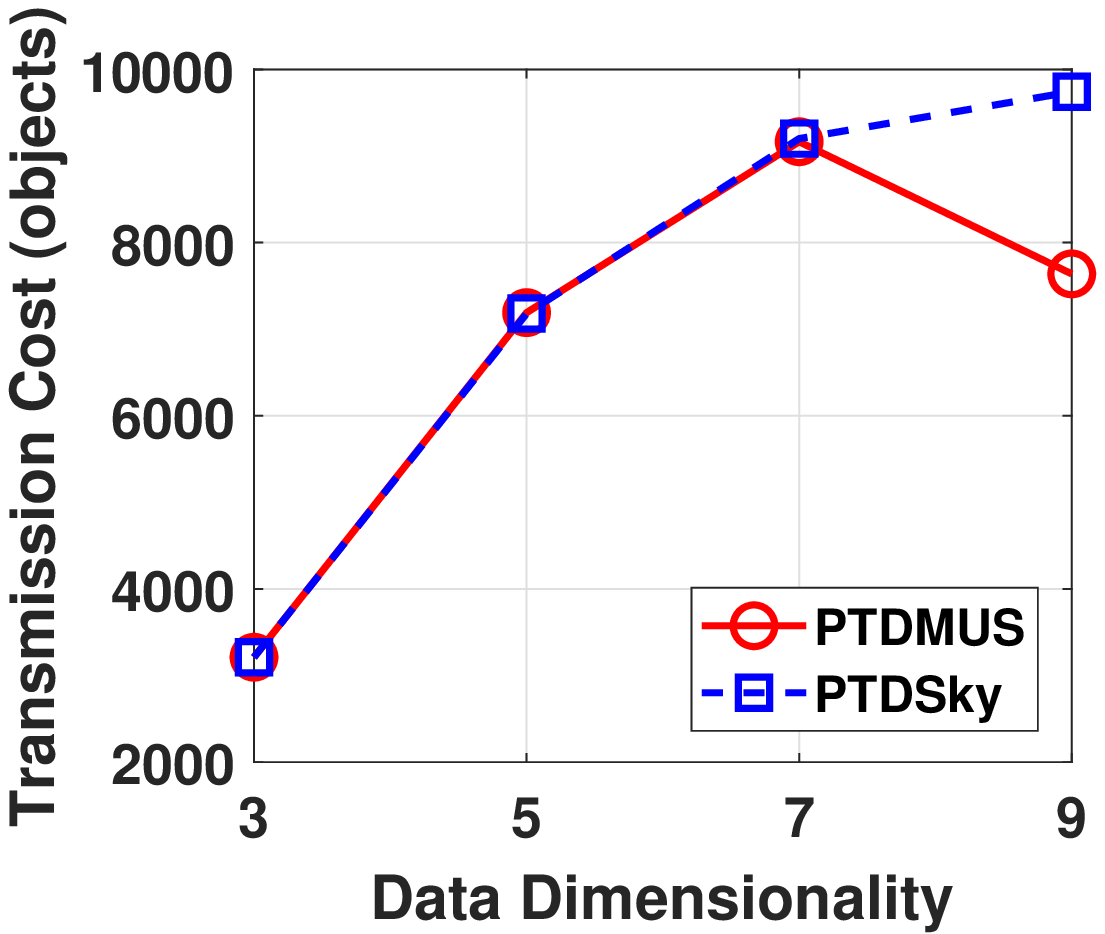}}%\hspace{.5in}
	\subfigure[Precision]{
		\label{fig:dimension:precision} %% label for 1st subfigure
		\includegraphics[width=0.24 \textwidth]{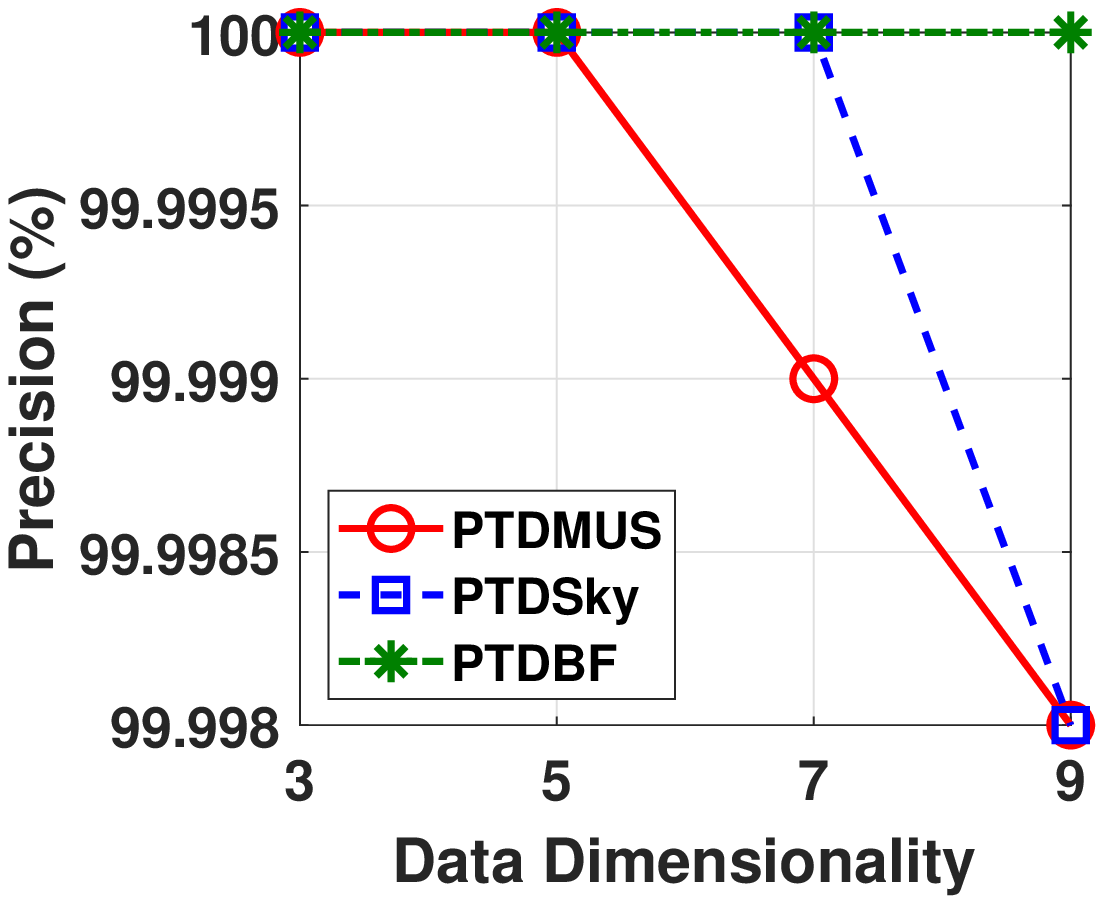}}%\hspace{.2in}
	\subfigure[Recall]{
		\label{fig:dimension:recall} %% label for 2nd subfigure
		\includegraphics[width=0.24 \textwidth]{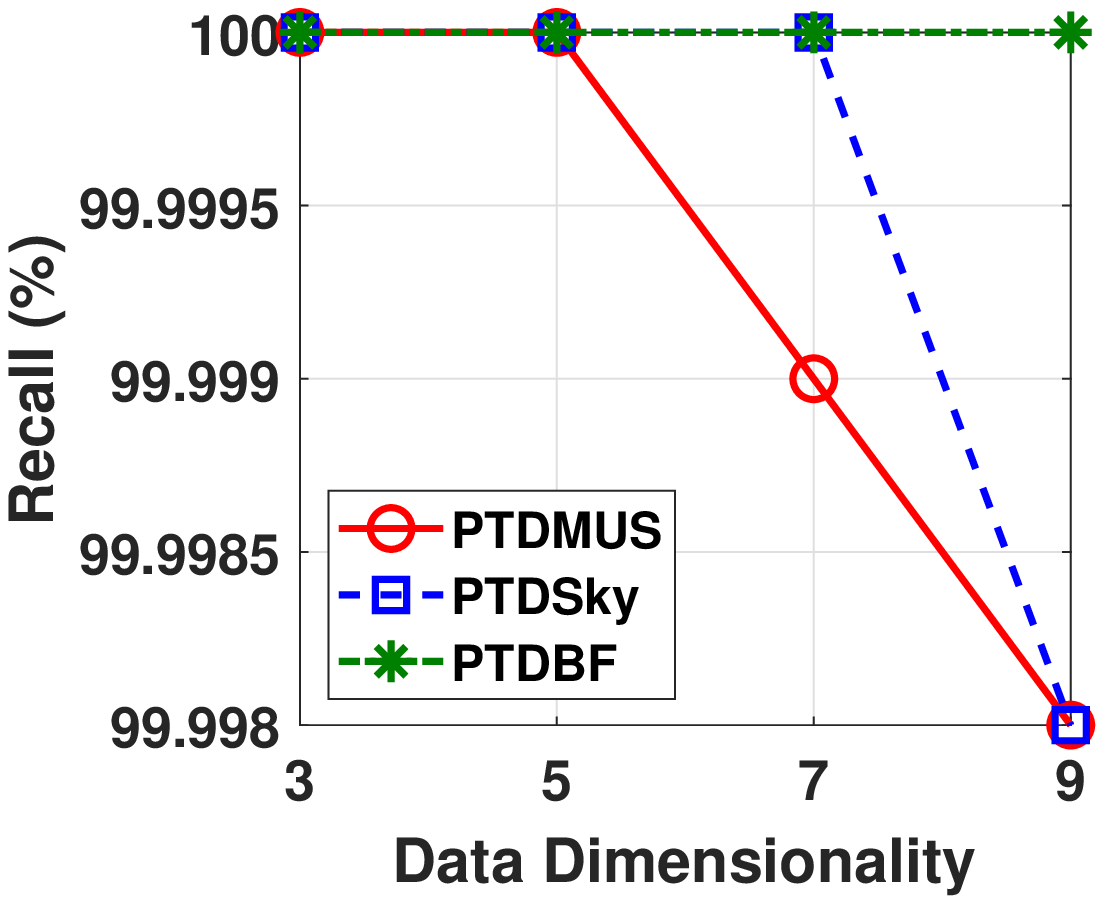}}%\hspace{.5in}
	\caption{Effect of data dimensionality $d$ on \subref{fig:dimension:time} Computation Time, \subref{fig:dimension:transmission} Transmission Cost, \subref{fig:dimension:precision} Precision and \subref{fig:dimension:recall} Recall.
	}
	\label{fig:dimension} %% label for entire figure
	\vspace{-10pt}
\end{figure*}

\subsection{Threshold}
As shown in Fig.~\ref{fig:threshold:time}, the computation time of PTDSky grows linearly as the given threshold $\delta$ increases. This is because PTDSky needs to process more candidate objects for the threshold-based $k$-skyband when the data dimension is high ($d=9$), and when the threshold $\delta$ becomes loose, the computation time of PTDSky increases. PTDBF just needs ro handle all the uncertain data objects and sorts the result by the dominant score of object directly, thus having the worst computation time which is irrelevant to the threshold. PTDMUS has the best performance on computation time since it can avoid unnecessary computation on irrelevant objects with the minimum checking time. PTDMUS can perform almost 10 times faster than PTDSky when $\delta=50$.

According to Fig.~\ref{fig:threshold:transmission}, PTDMUS can save almost 30\% transmission cost comparing to PTDSky when $\delta=50$. In general, both PTDSky and PTDMUS need higher transmission cost since the local and the global candidate sets become large as $\delta$ increases. However, with a table recording the minimum checking times of possible candidates, the monitor and coordinator nodes in PTDMUS do not need to exchange the information of candidate sets too much if the continuous query result does not change a lot. As a result, PTDMUS can outperform PTDSky significantly. In addition, when $\delta$ increases, the global candidate set on the coordinator node becomes larger. Then PTDMUS can record more information (minimum checking times) of candidate objects, thereby avoiding the unnecessary transmission of irrelevant objects. In the rest of simulation, we choose $\delta=30$ as the default threshold. Note that PTDBF performs in a centralized way, so it doesn't have transmission cost.

Since we use threshold-based $k$-skyband to prune irrelevant objects in our proposed approach, we now measure its influence on the accuracy of result for the query. Fig.~\ref{fig:threshold:precision} and Fig.~\ref{fig:threshold:recall} show that PTDMUS only loses less than 0.001\% performance on accuracy and recall respectively. Such a tiny performance gap can be recognized as a tolerant error. In other words, with the minimum checking time, PTDMUS can reduce transmission cost significantly with good accuracy and recall in the meantime.

\subsection{Data Dimensionality}
Fig.~\ref{fig:dimension:time} shows that PTDBF has poor performance on computation time, especially when the dimension $d$ is small. In general, for each data object, the number of its dominated objects is large when $d$ is small. In other words, each data object has a high probability to be dominated by the other objects when $d$ is small. Hence, PTDBF needs more computations on dominance checks when $d$ is small.
%The reason is that PTDBF cannot prune irrelevant uncertain data objects effectively for small $d$ and the large candidate set. Thus, PTDBF significantly increases the processing time on dominance checks.
In addition, from the implementation perspective, PTDBF needs much more branch operations (conditions) for the dominance checks between each pair of objects, so its computation time is the worst. In comparison with PTDBF, both $k$-skyband based methods, PTDSky and PTDMUS, need less computation time since $k$-skyband can utilize the characteristics of R-trees and MBRs for precluding irrelevant objects effectively. When the dimensionality becomes large ($d=9$), PTDSky and PTDBF have similar performance in transmission time. On the other hand, PTDMUS has the best computation time and outperforms PTDSky and PTDBF by more than 85\% when $d=9$.

As Fig.~\ref{fig:dimension:transmission} shown, PTDSky and PTDMUS have very similar performance in average computation cost when $d\leq 7$. In this simulation, the size of the given uncertain data set $|U|$ is 10,000. We can observe that PTDSky and PTDMUS need to transmit more than 9,000 candidate objects when $d\leq 7$. Such a phenomenon indicates that the score of an object dominating another objects decreases significantly, so the number of candidate objects becomes large and near to $|U|$. In the case of $d>7$, the coordinator node in PTDMUS can record the minimum checking time of more than 9,000 objects and the minimum checking time table can help coordinator node avoid transmitting the information of irrelevant objects to monitor nodes at some time slots (runs). Hence, the average transmission cost of PTDMUS can be improved nearly 20\% in such a scenario ($d=9$).

Even using a predictive mechanism to reduce the frequency of updating candidate objects, Fig.~\ref{fig:dimension:precision} and Fig.~\ref{fig:dimension:recall} show that PTDMUS can achieve almost the same performance on precision and recall as PTDSky does. In comparison with PTDSky, PTDMUS only loses 0.001\% performances on both precision and recall for $d=7$. In most applications, such a tiny lose of performance can be recognized as a tolerant error.

\begin{figure*}[t]
	\centering
	\subfigure[Computation Time]{
		\label{fig:number_monitor_node:time} %% label for 1st subfigure
		\includegraphics[width=0.24 \textwidth]{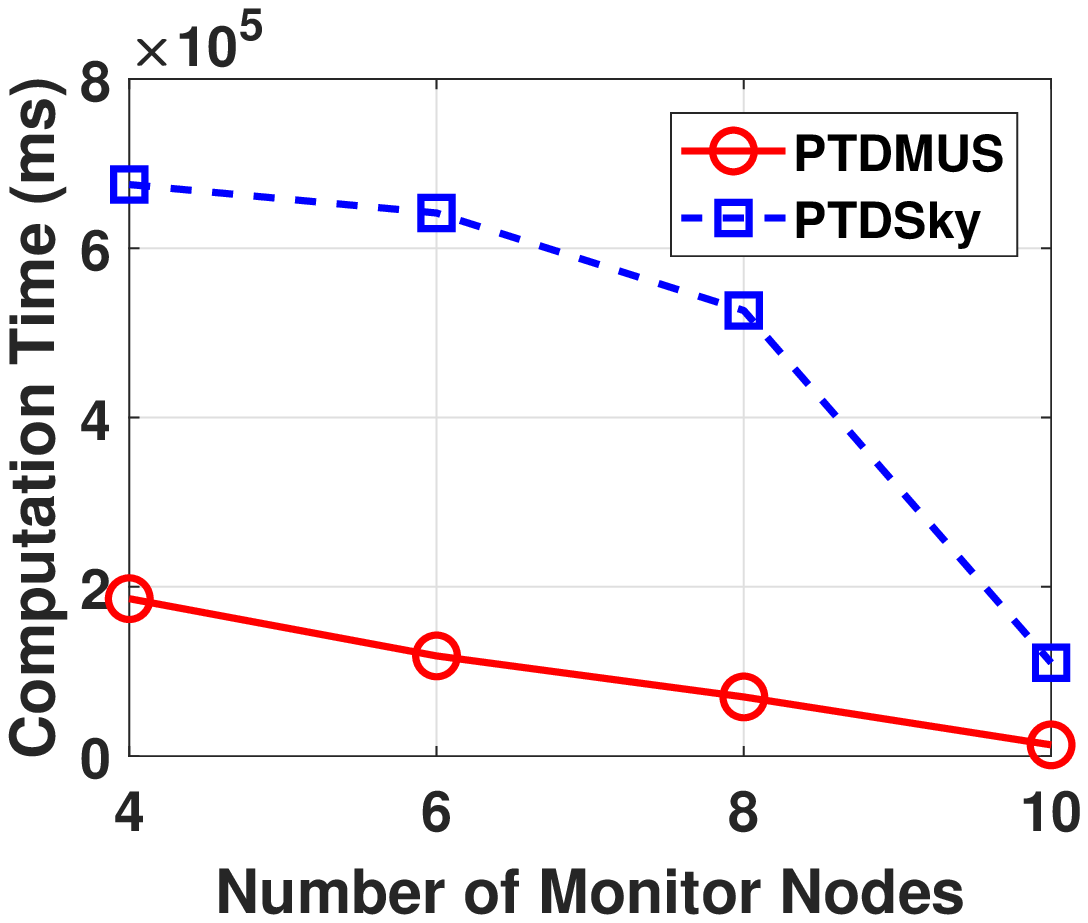}}%\hspace{.05in}
	\subfigure[Transmission Cost]{
		\label{fig:number_monitor_node:transmission} %% label for 2nd subfigure
		\includegraphics[width=0.24 \textwidth]{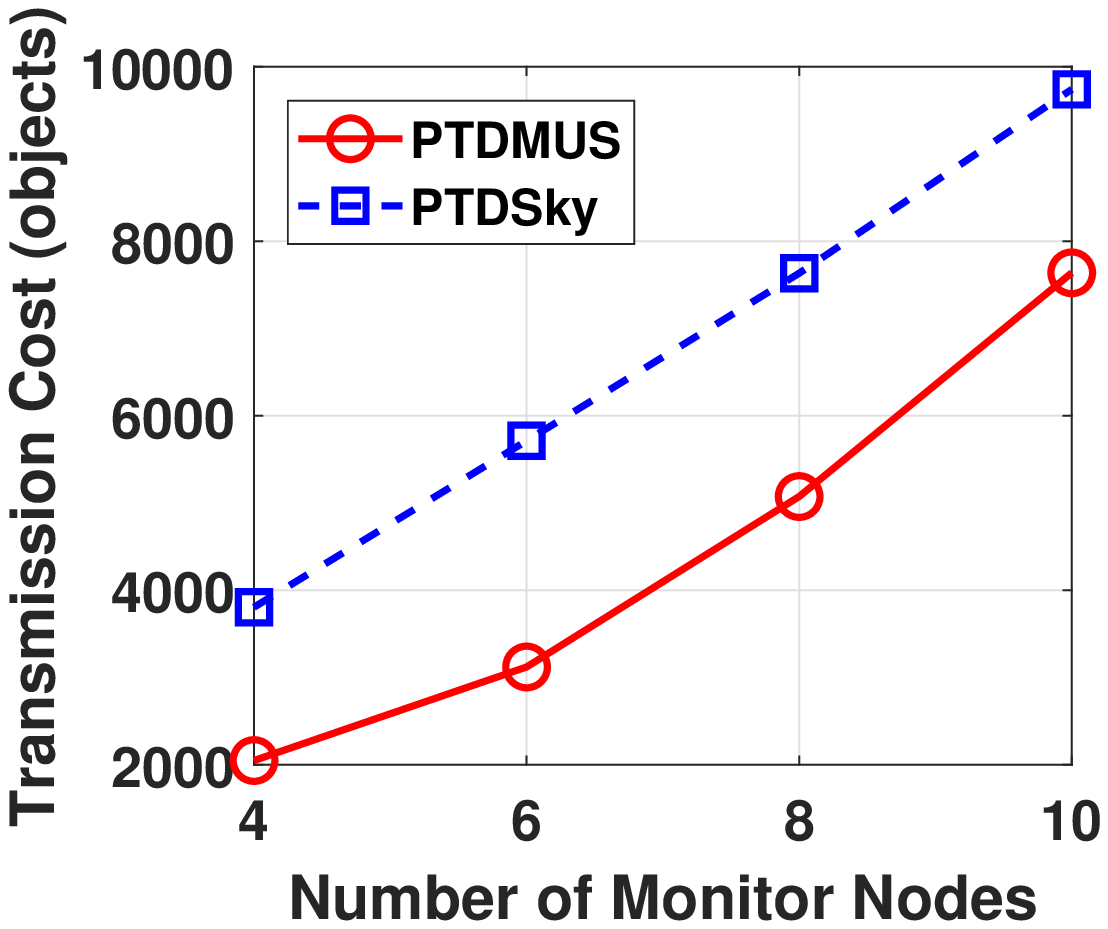}}%\hspace{.5in}
	\subfigure[Precision]{
		\label{fig:number_monitor_node:precision} %% label for 1st subfigure
		\includegraphics[width=0.24 \textwidth]{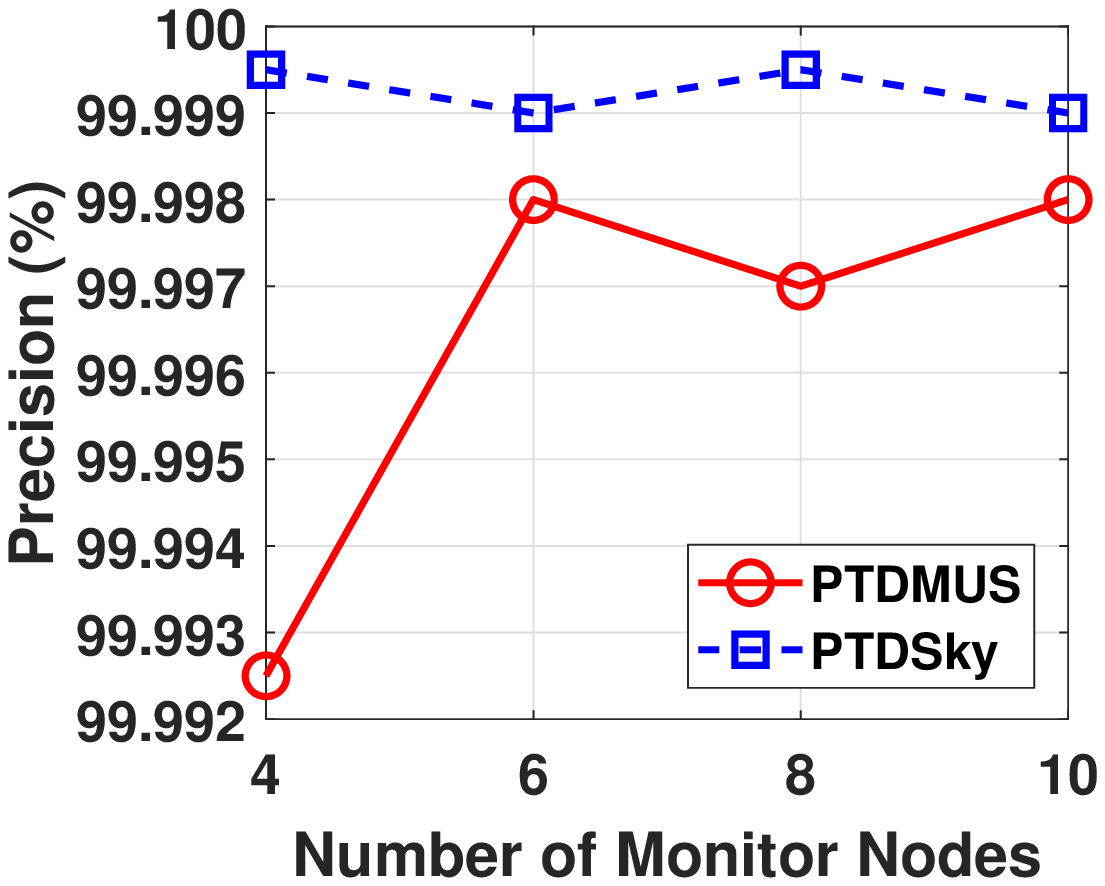}}%\hspace{.2in}
	\subfigure[Recall]{
		\label{fig:number_monitor_node:recall} %% label for 2nd subfigure
		\includegraphics[width=0.24 \textwidth]{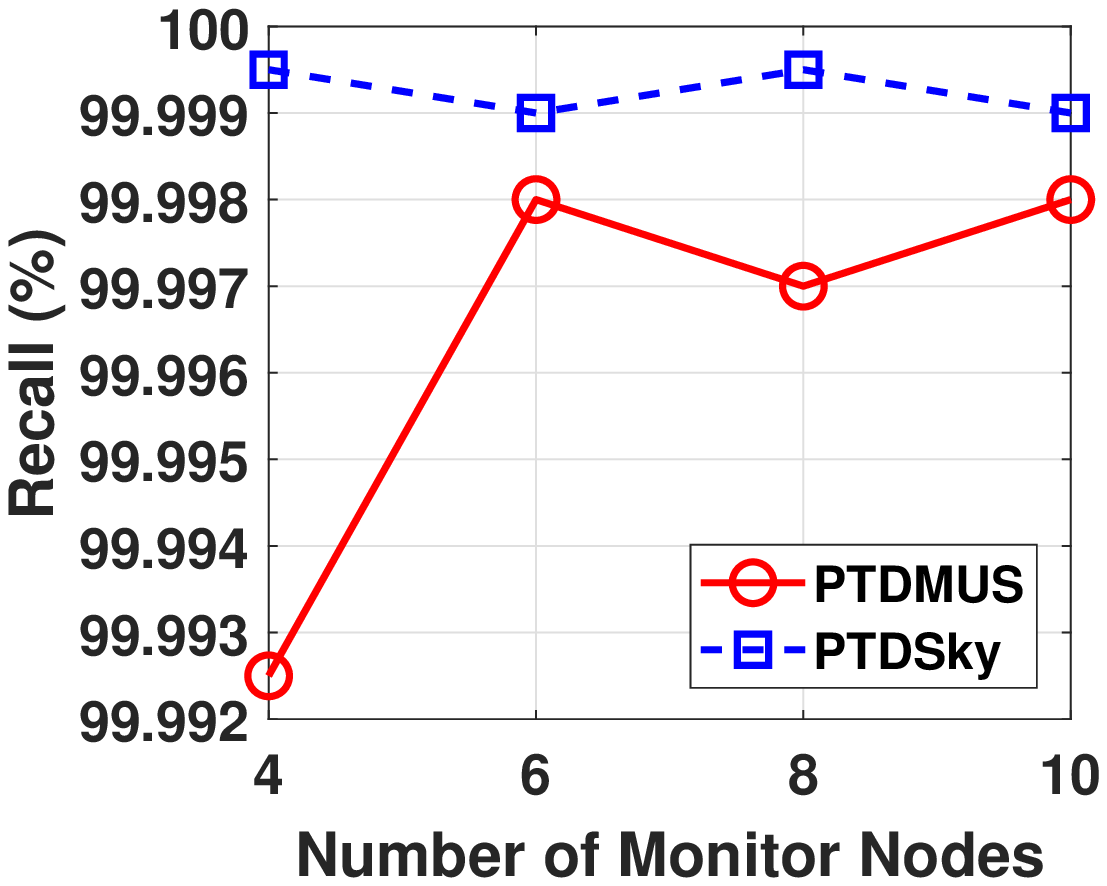}}%\hspace{.5in}
	\caption{Effect of number of monitor nodes $m$ on \subref{fig:number_monitor_node:time} Computation Time, \subref{fig:number_monitor_node:transmission} Transmission Cost, \subref{fig:number_monitor_node:precision} Precision and \subref{fig:number_monitor_node:recall} Recall.
	}
	\label{fig:number_monitor_node} %% label for entire figure
	\vspace{-10pt}
\end{figure*}
\begin{figure*}[t]
	\centering
	\subfigure[Computation Time]{
		\label{fig:sliding_windows:time} %% label for 1st subfigure
		\includegraphics[width=0.24 \textwidth]{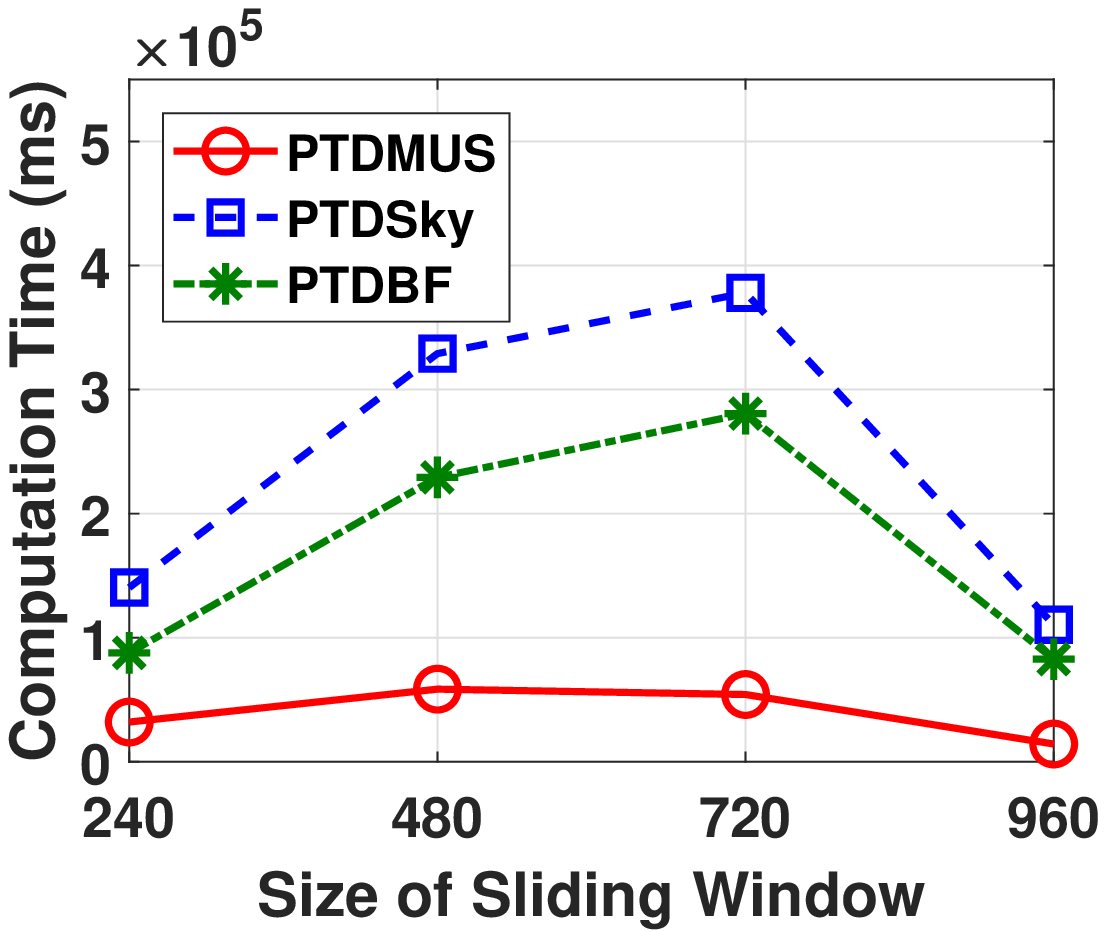}}%\hspace{.05in}
	\subfigure[Transmission Cost]{
		\label{fig:sliding_windows:transmission} %% label for 2nd subfigure
		\includegraphics[width=0.24 \textwidth]{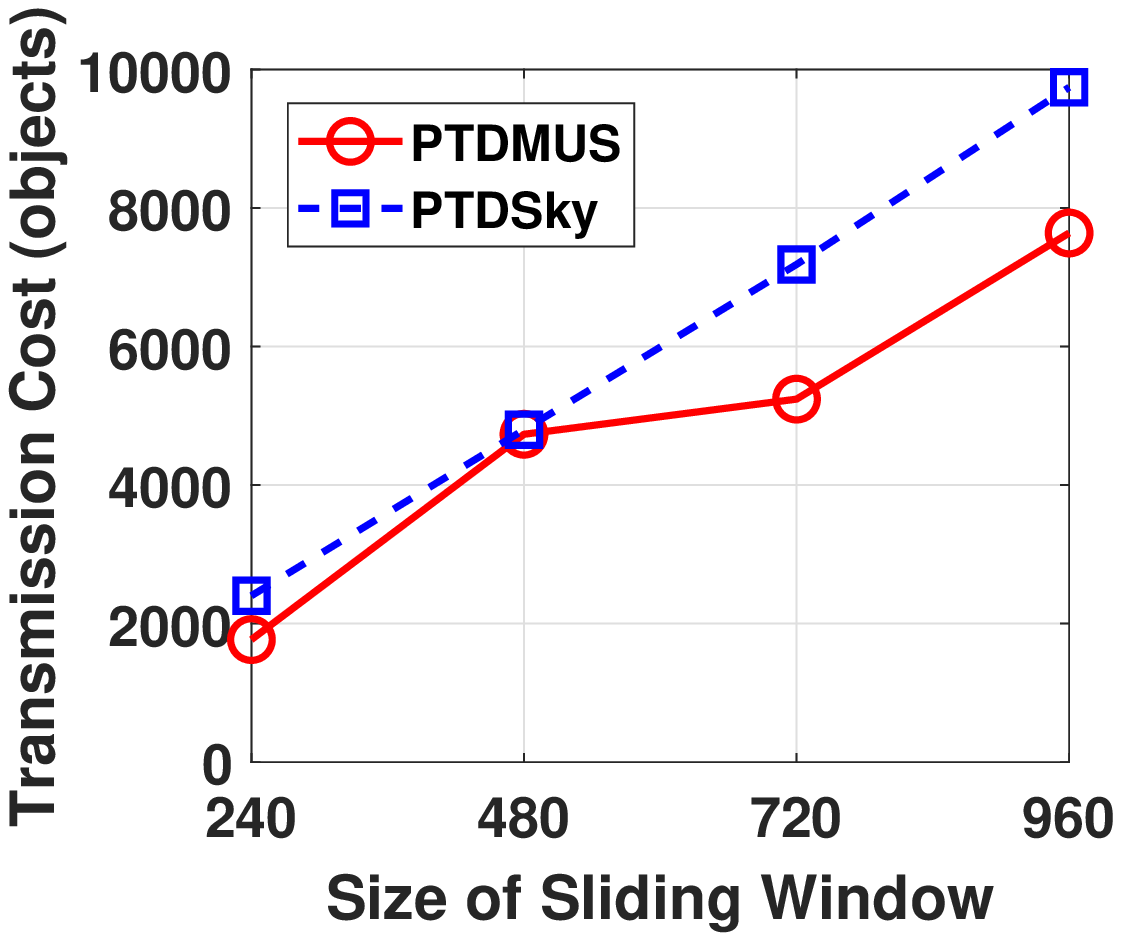}}%\hspace{.5in}
	\subfigure[Precision]{
		\label{fig:sliding_windows:precision} %% label for 1st subfigure
		\includegraphics[width=0.24 \textwidth]{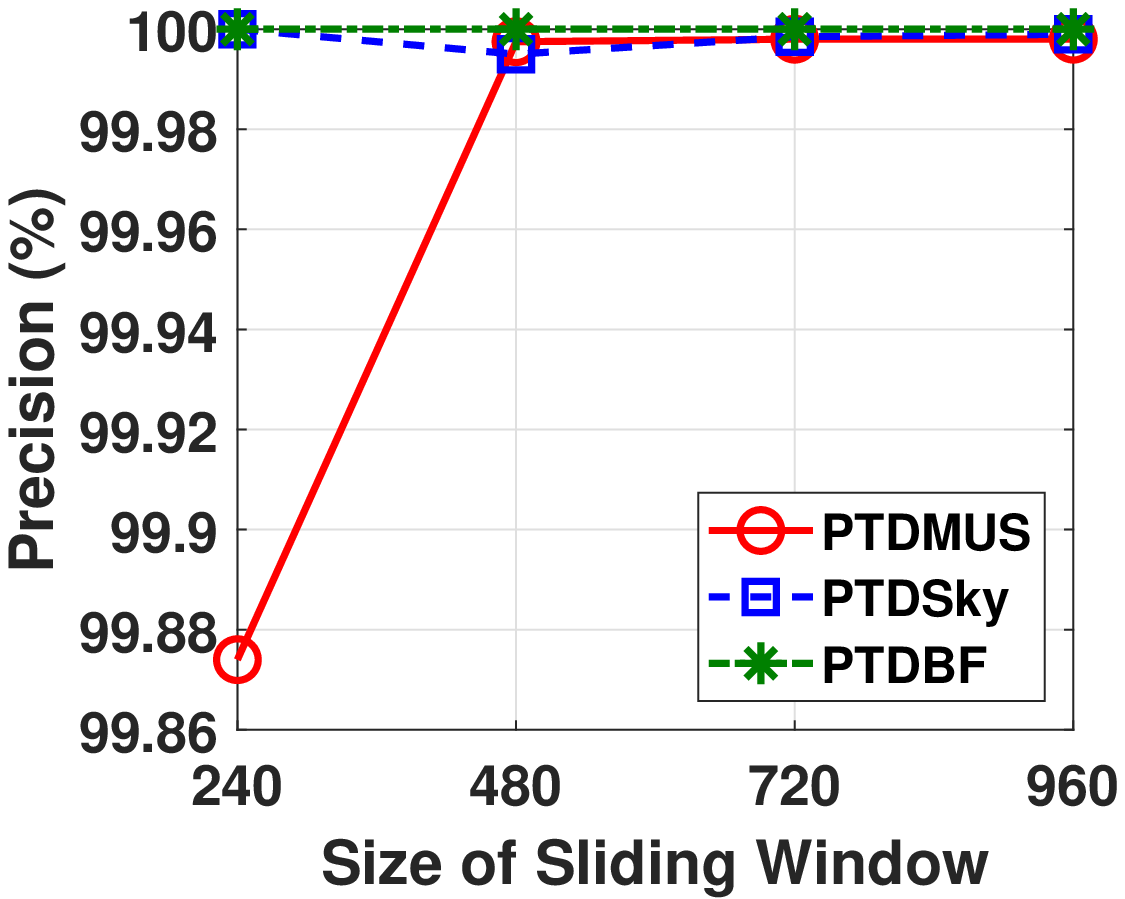}}%\hspace{.2in}
	\subfigure[Recall]{
		\label{fig:sliding_windows:recall} %% label for 2nd subfigure
		\includegraphics[width=0.24 \textwidth]{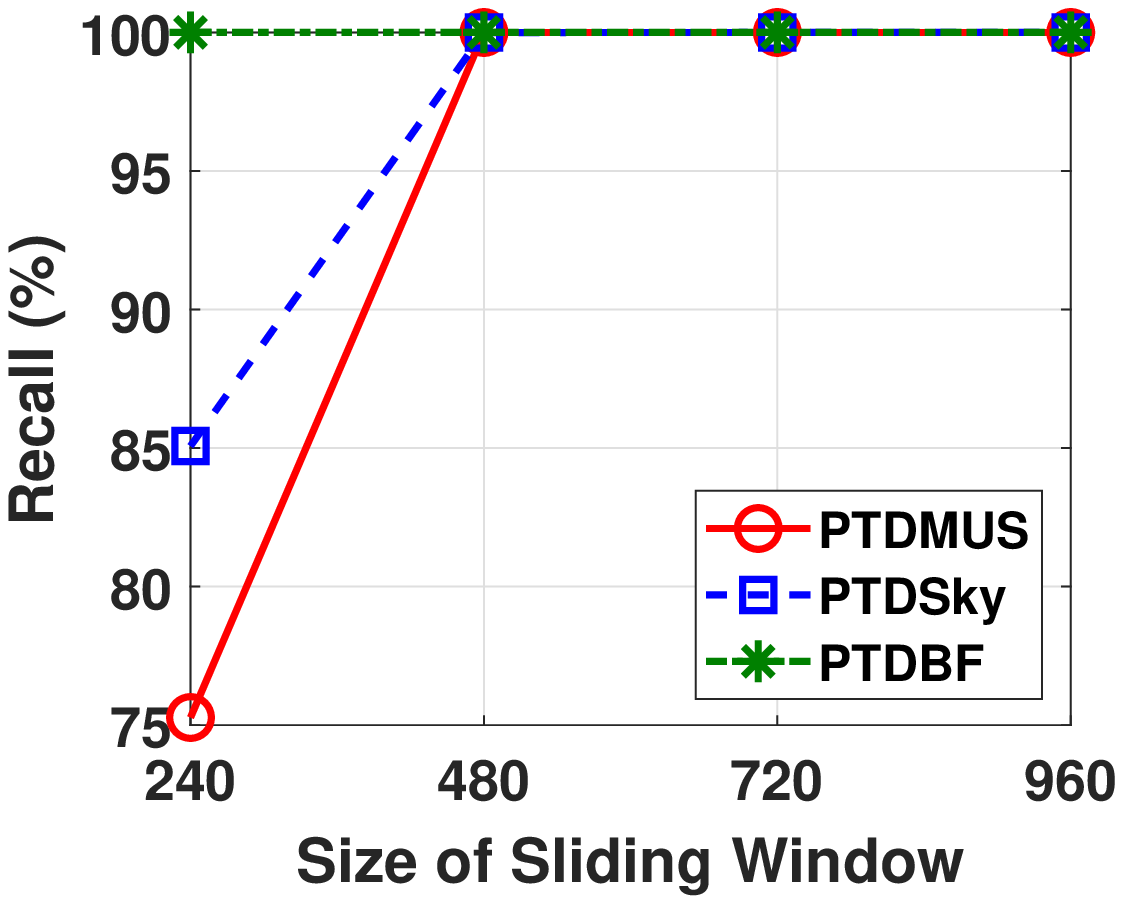}}%\hspace{.5in}
	\caption{Effect of size of sliding window $|SW_j|$ on \subref{fig:sliding_windows:time} Computation Time, \subref{fig:sliding_windows:transmission} Transmission Cost, \subref{fig:sliding_windows:precision} Precision and \subref{fig:sliding_windows:recall} Recall.
	}
	\label{fig:sliding_windows} %% label for entire figure
	\vspace{-10pt}
\end{figure*}
\begin{figure*}[t]
	\centering
	\subfigure[Computation Time]{
		\label{fig:k:time} %% label for 1st subfigure
		\includegraphics[width=0.24 \textwidth]{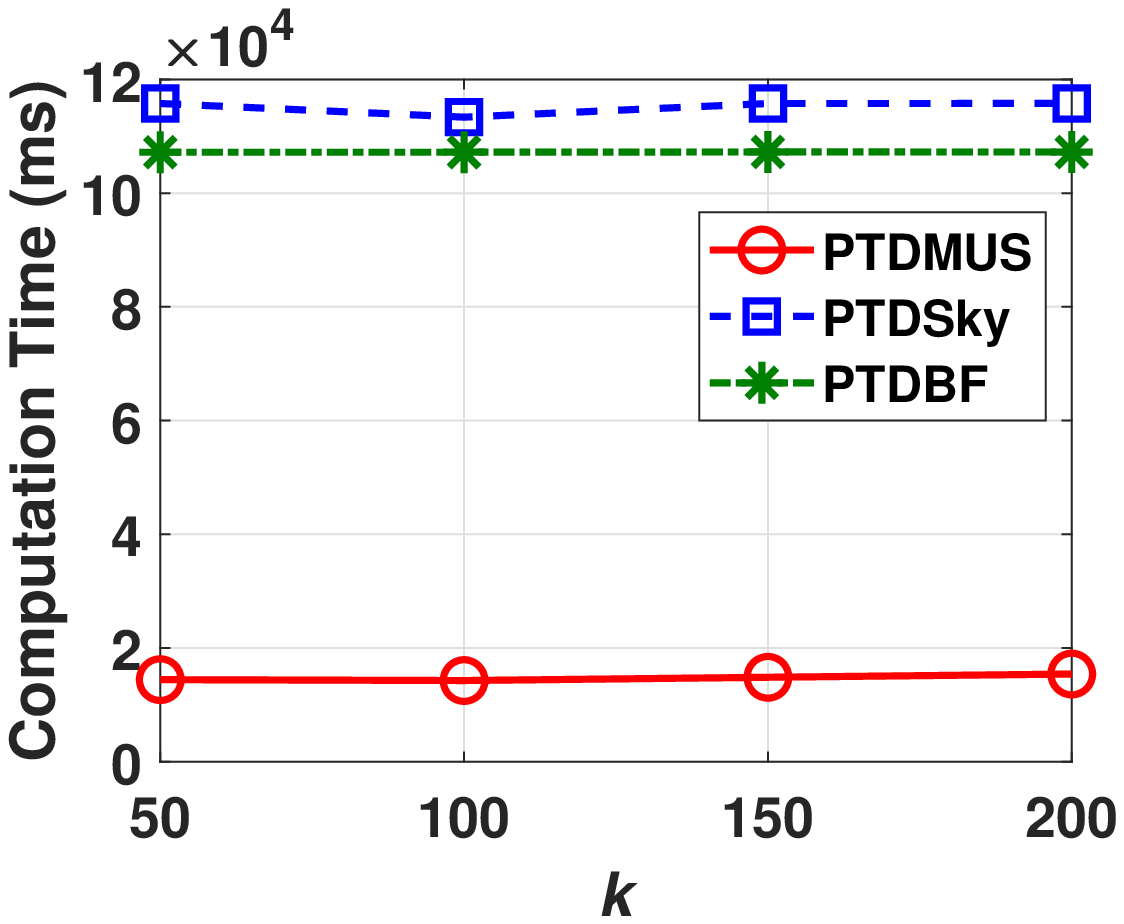}}%\hspace{.05in}
	\subfigure[Transmission Cost]{
		\label{fig:k:transmission} %% label for 2nd subfigure
		\includegraphics[width=0.24 \textwidth]{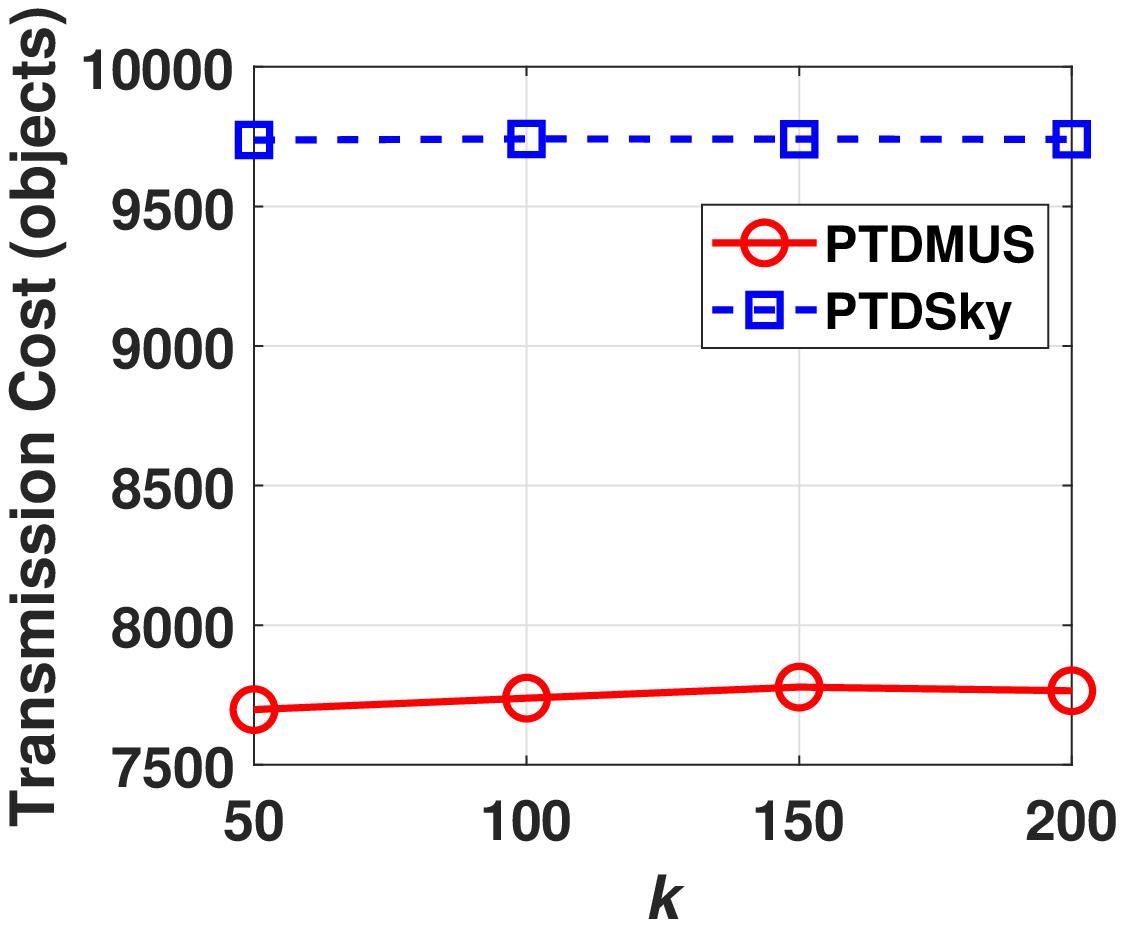}}%\hspace{.5in}
	\subfigure[Precision]{
		\label{fig:k:precision} %% label for 1st subfigure
		\includegraphics[width=0.24 \textwidth]{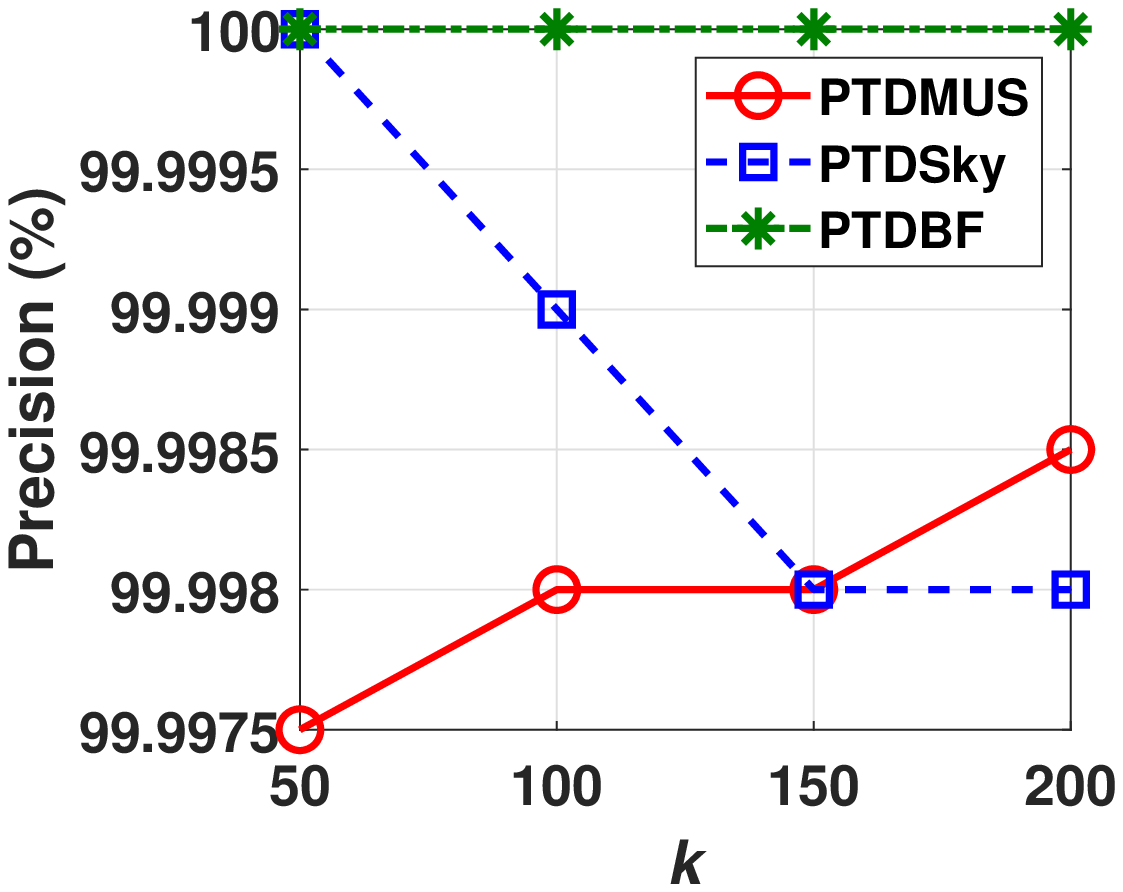}}%\hspace{.2in}
	\subfigure[Recall]{
		\label{fig:k:recall} %% label for 2nd subfigure
		\includegraphics[width=0.24 \textwidth]{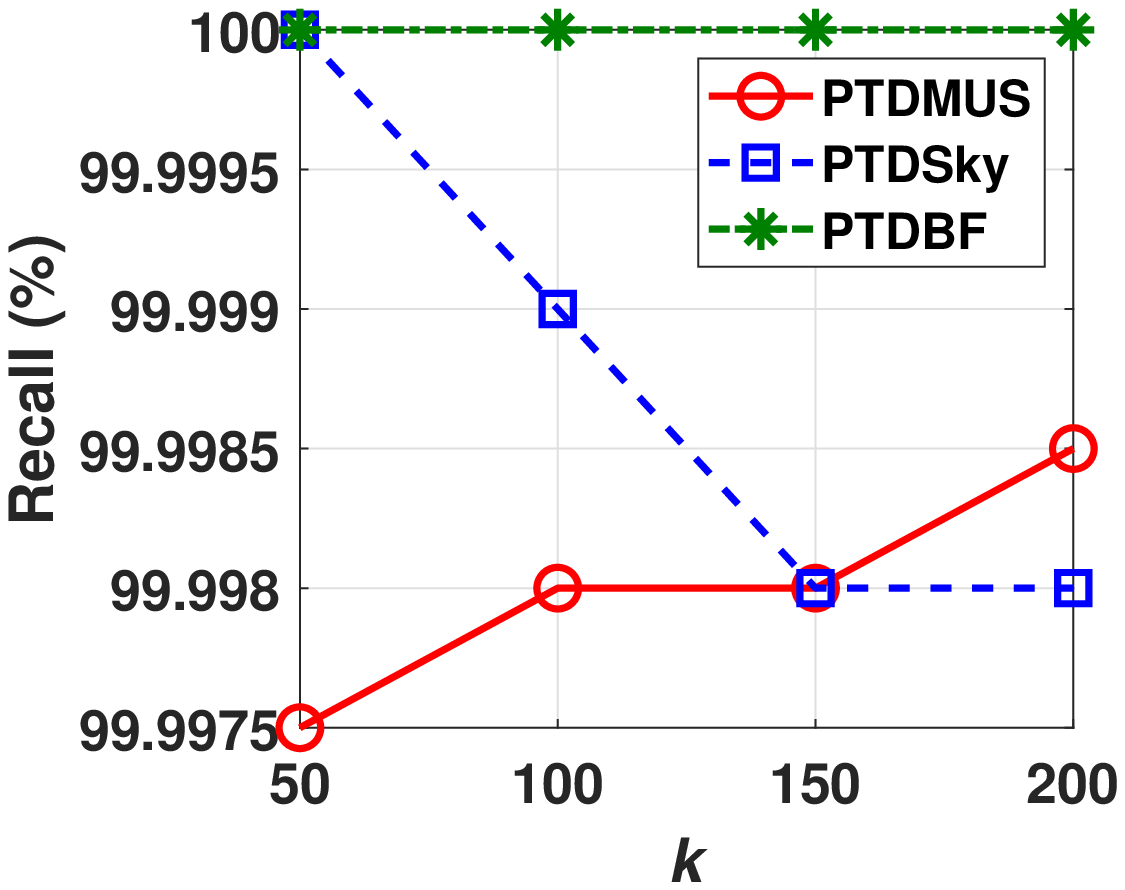}}%\hspace{.5in}
	\caption{Effect of value of $k$ on \subref{fig:k:time} Computation Time, \subref{fig:k:transmission} Transmission Cost, \subref{fig:k:precision} Precision and \subref{fig:k:recall} Recall.
	}
	\label{fig:k} %% label for entire figure
	\vspace{-10pt}
\end{figure*}

\subsection{Number of Monitor Nodes}
The considered environment is implemented in a parallel model. We now discuss the performance of each method in different scenarios with various numbers of monitor nodes. Note that there are no results of PTDBF in this part since PTDBF is a centralized method.
Fig.~\ref{fig:number_monitor_node:time} shows that both PTDMUS and PTDSky need less computation time if the number of monitor nodes increases. With the minimum checking time mechanism, PTDMUS precludes irrelevant objects more effectively than PTDSky does. Thus, PTDMUS outperforms PTDSky in computation time by almost 60\%. Fig.~\ref{fig:number_monitor_node:transmission} indicates that the average total transmission costs between monitor and coordinator nodes in PTDMUS and PTDSky are increasing as monitor nodes become more. In this simulation, we fix the size of sliding window in each monitor node, so the transmission cost is related to $m\times |SW_j|$. With the minimum checking time table, PTDMUS can save about more than 2,000 transmission cost (objects) under various number of monitor nodes.

According to the simulation results in Fig.~\ref{fig:number_monitor_node:precision} and Fig.~\ref{fig:number_monitor_node:recall}, PTDMUS is only $0.01\%$ worse than PTDSky on both accuracy and recall in different scenarios with different number of monitor nodes. Again, such a tiny performance gap can be recognized as a tolerant error for most applications. In summary, in comparison with PTDSky, PTDMUS reduces the average computation time and the transmission cost significantly while maintaining nearly identical accuracy and recall.

\subsection{Size of Sliding Window}
In this part, we discuss the effect of different sizes of sliding windows $|SW_j|$ on monitor nodes. In Fig.~\ref{fig:sliding_windows:time}, it shows that all the methods have better computation time performance when $|SW_j|$ is relatively small ($|SW_j|=240$) or large ($|SW_j|=960$). For each monitor node, the sliding window can be recognized as a buffer and it is used to save the data objects that need to be processed. In general, the computation cost will increase as the size of sliding window becomes big. It was shown in Fig.~\ref{fig:sliding_windows:time} that the computation time of all the methods is significantly reduced if $|SW_j|=960$. The reason is that the coordinator node records a large candidate set and the upper bound of its size is $m\times|SW_j|$. In the case of $|SW_j|=960$, the coordinator node will record 9,600 data objects at most and it approaches to the size of the given data set $|U|=10,000$. According to~\eqref{eq:monitor_time}, all the methods only need to execute $\varDelta t=40$ runs (slots) and each monitor node deals with only one new input object at every time slot. In such a scenario, the score for the new input object to become the candidate object is very low. Thus, the computation cost of updating the global candidate on the coordinator node also significantly decreases. In the case of $|SW_j|=240$, the reason for all the compared methods to have a good computation cost is that the small $|SW_j|$ makes the computation time of each run (slot) very fast. In summary, PTDMUS has the best performance on computation cost in all the considered cases with different sizes of sliding windows. As shown in Fig.~\ref{fig:sliding_windows:transmission}, PTDMUS needs lower transmission cost than PTDSky does in all the scenarios with different sizes of sliding windows. PTDSky has a similar performance to PTDMUS in transmission cost only when $|SW_j|=480$ but PTDMUS is still better.

Fig.~\ref{fig:sliding_windows:precision} and Fig.~\ref{fig:sliding_windows:recall} show that both PTDMUS and PTDSky achieves 99.998\% precision and recall when $|SW_j|$ is 720 or 960. If $|SW_j|=480$, the performance gap between PTDMUS and PTDSky is smaller than 0.01\% in terms of precision and recall. If $|SW_j|=240$, PTDMUS loses about 0.126\% precision and 9.8\% recall in comparison with PTDSky. However, such a high precision (99.87\%) performance provided by PTDMUS is still allowable for most applications except for financial and emergency services.
\vspace{-5pt}

\begin{figure*}[t]
	\centering
	\subfigure[Computation Time]{
		\label{fig:uncertainty:time} %% label for 1st subfigure
		\includegraphics[width=0.24 \textwidth]{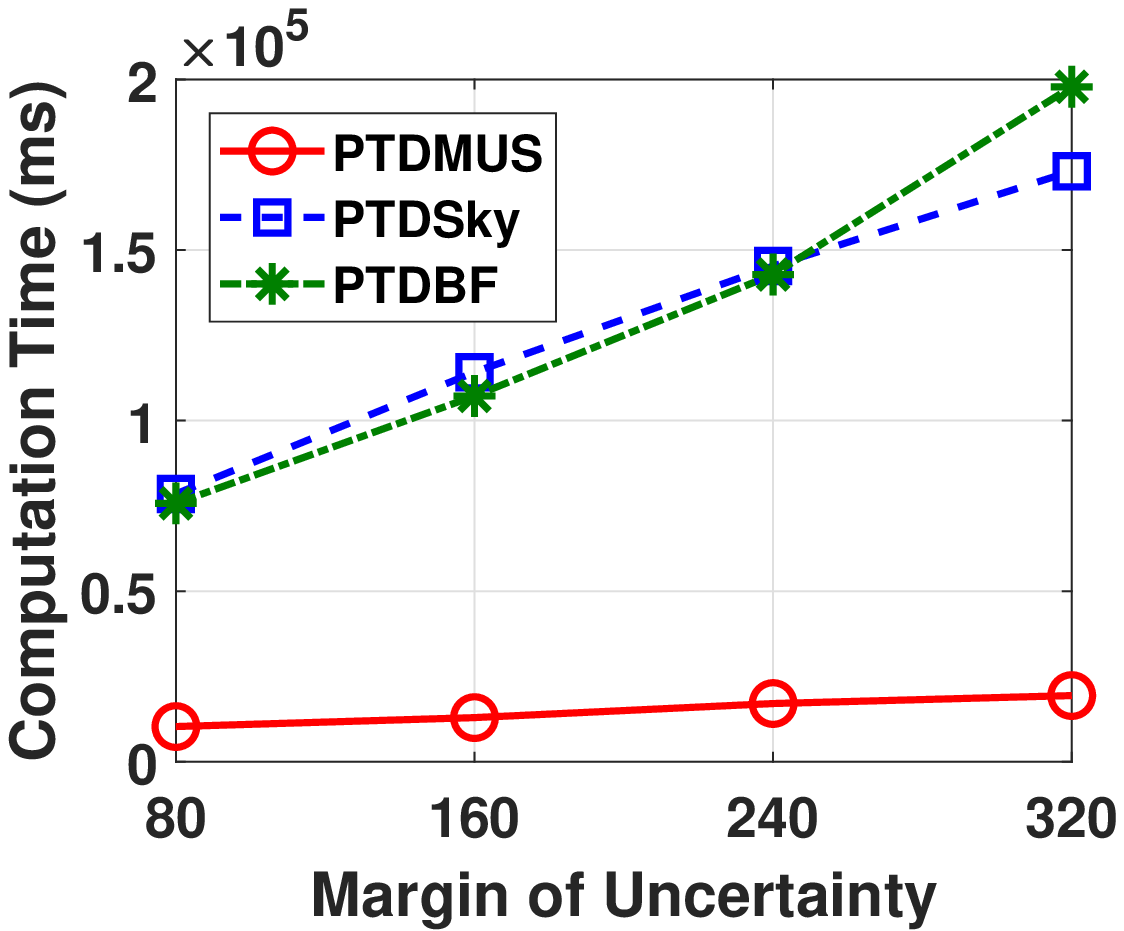}}%\hspace{.05in}
	\subfigure[Transmission Cost]{
		\label{fig:uncertainty:transmission} %% label for 2nd subfigure
		\includegraphics[width=0.24 \textwidth]{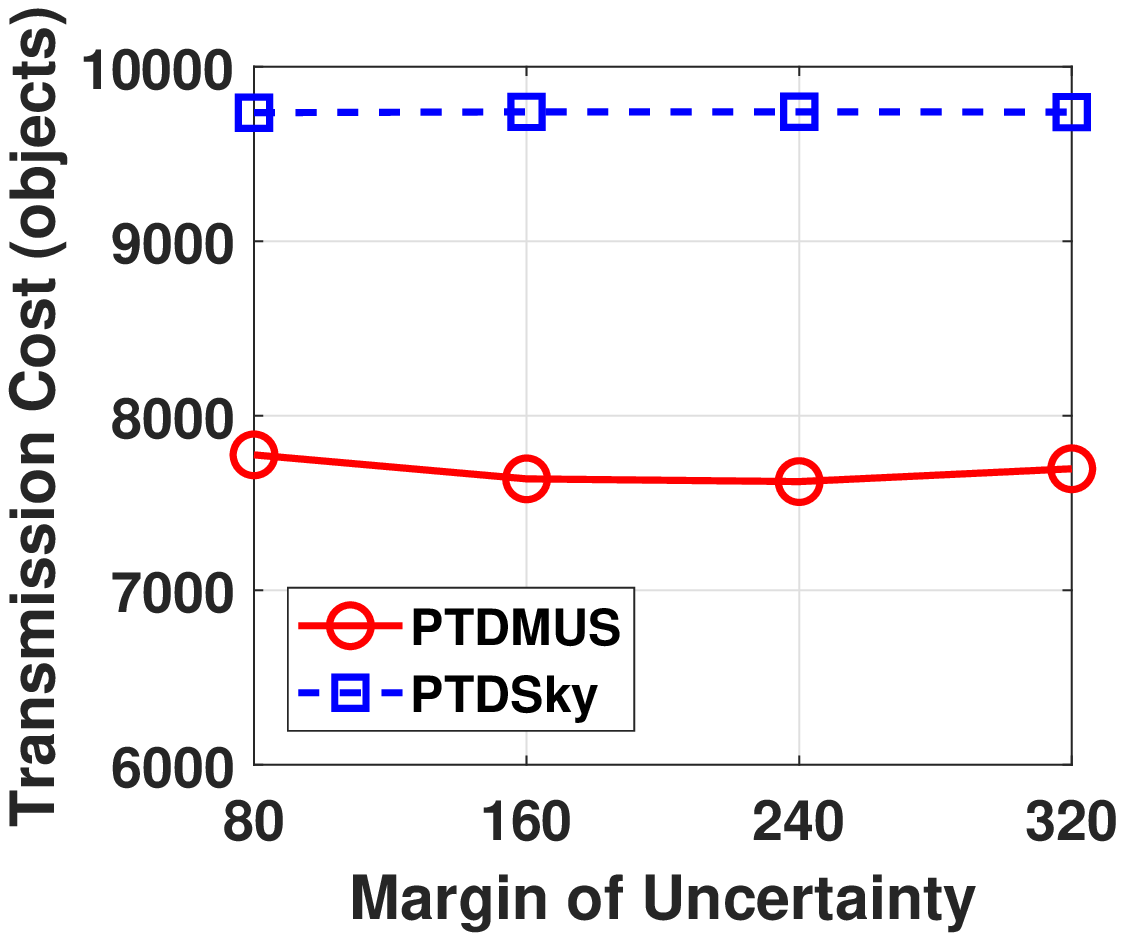}}%\hspace{.5in}
	\subfigure[Precision]{
		\label{fig:uncertainty:precision} %% label for 1st subfigure
		\includegraphics[width=0.24 \textwidth]{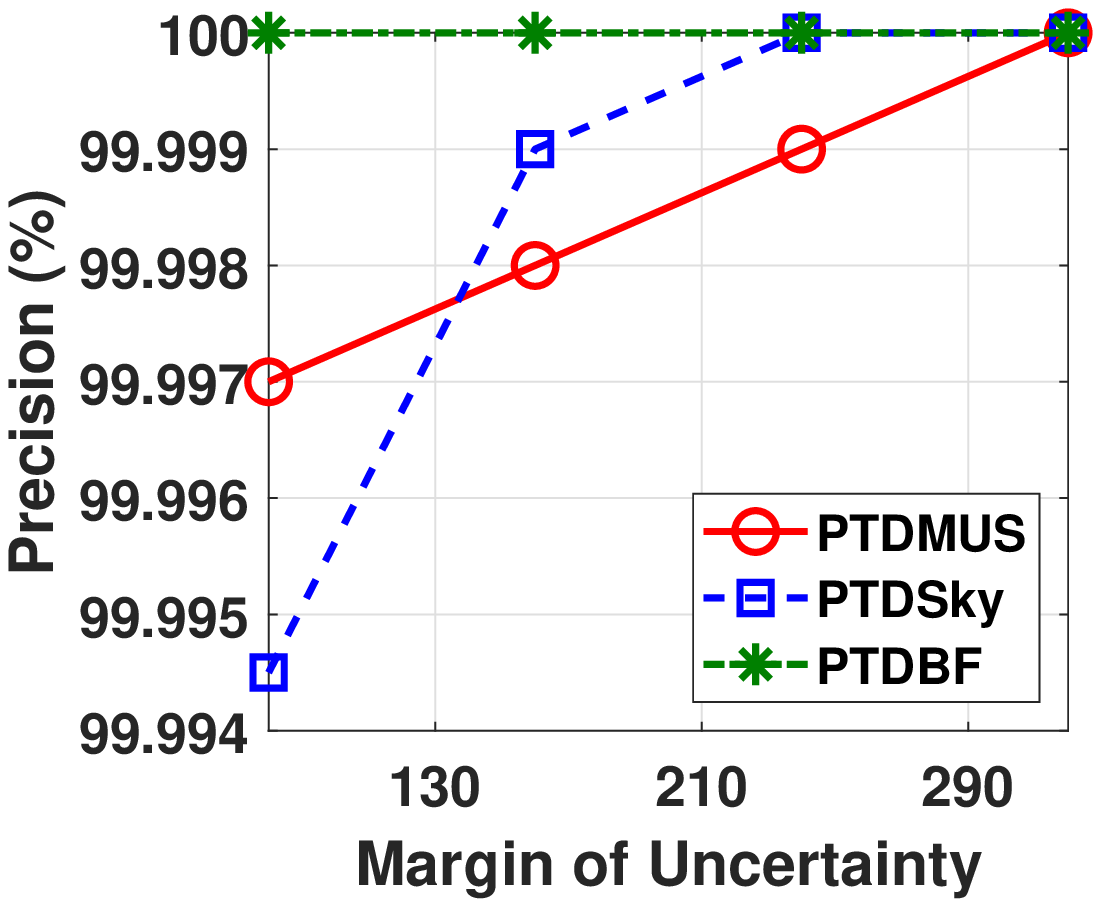}}%\hspace{.2in}
	\subfigure[Recall]{
		\label{fig:uncertainty:recall} %% label for 2nd subfigure
		\includegraphics[width=0.24 \textwidth]{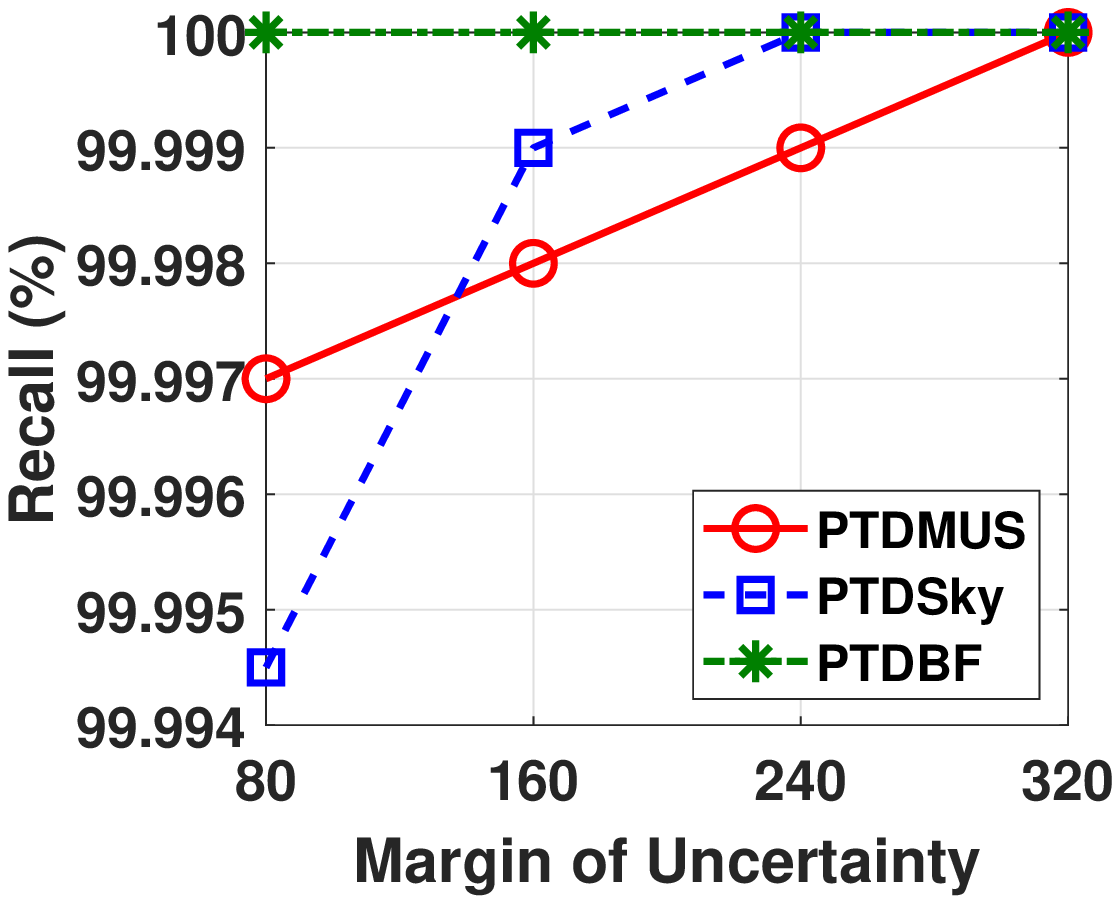}}%\hspace{.5in}
	\caption{Effect of margin of uncertainty $M$ on \subref{fig:uncertainty:time} Computation Time, \subref{fig:uncertainty:transmission} Transmission Cost, \subref{fig:uncertainty:precision} Precision and \subref{fig:uncertainty:recall} Recall.
	}
	\label{fig:uncertainty} %% label for entire figure
	\vspace{-5pt}
\end{figure*}

\subsection{Value of \textit{k}}
Since we consider the top-$k$ dominating query, various values of $k$ may affect the performance. Fig.~\ref{fig:k:time} shows that the performance of all the methods in terms of computation time are independent of the value of $k$. In the case of a high dimensional data set ($d=9$), PTDSky is slightly worse than PTDBF since the coordinator node wastes too much computation time in computing the global probabilistic threshold-based $k$-skyband with too many irrelevant objects. Such a similar result has been presented in Fig.~\ref{fig:dimension:time}. Conversely, PTDMUS improves more than 80\% computation time comparing to PTDSky and PTDBF. From Fig.~\ref{fig:k:transmission}, we can observe that the value of $k$ is also independent of the transmission cost for both PTDMUS and PTDSky. The reason is that both PTDMUS and PTDSky use a threshold $\delta$ to preclude the irrelevant objects, where $\delta$ is much smaller than $k$. In addition, PTDMUS can save almost 20\% transmission cost due to the usage of the minimum checking time.

Fig.~\ref{fig:k:precision} and Fig.~\ref{fig:k:recall} show that PTDMUS has the same trend in precision and recall. The precision and recall of PTDMUS slightly increase as the value of $k$ increases. When $k=200$, PTDMUS can achieve $99.9985\%$ precision and recall. Although PTDSky also has the same precision and recall, PTDSky performs better than PTDMUS in precision and recall as the value of $k$ becomes smaller. PTDSky achieves 100\% precision and recall when $k\leq 50$. If $k>150$, PTDMUS will achieve better precision and recall than PTDSky does.

\subsection{Margin of Uncertainty}
We last discuss the effect of object's margin of uncertainty $M$, which is also called the object size. In general, the MBR of an uncertain data object becomes large as $M$ increases and thus the occurrence of partial dominance will increase. As shown in Fig.~\ref{fig:uncertainty:time}, PTDMUS has the best computation time performance and with 60\% improvement in comparison with PTDSky and PTDBF. PTDSky and PTDBF have similar computation time for $M\leq 240$. When $M>240$, PTDBF becomes the worst one due to the large occurrence of partial dominants.
In Fig.~\ref{fig:uncertainty:transmission}, it is shown that the margin of uncertainty $M$ is independent to the size of the candidate set. Thus, the transmission costs of PTDMUS and PTDSky are not affected by the margin of uncertainty.

According to the results in Fig.~\ref{fig:uncertainty:precision} and Fig.~\ref{fig:uncertainty:recall}, the precision and recall of the query result provided by PTDMUS linearly increase as the margin of uncertainty $M$ increases. On the other hand, the precision and recall of PTDSky's query result increase more significantly as $M$ becomes larger. PTDSky can provide the query result with higher precision and recall only if $ 160\leq M<320$. For $M\leq 80$, PTDMUS will has better precision and recall than PTDSky by more than 0.0025\%.

\section{Conclusion}
\label{sec:conclusion}
In this paper, we have presented a new approach for Probabilistic Top-$k$ Dominating query over Multiple Uncertain data Streams (PTDMUS) to improve the computation efficiency of probabilistic top-$k$ dominating query for Edge-IoT applications. With the parallelism, the monitor nodes use the value of $k$ and threshold-based probabilistic $k$-skyband to preclude most of the irrelevant objects in advance, thereby significantly reducing transmission cost. The coordinator node caches the temporary result and uses the proposed approach, minimum checking time, to reduce the frequency of computing the dominant score of each object in the cache table. Such a way can effectively minimize the computation time and incrementally update the result of the probabilistic top-$k$ dominating query with less update frequency. The simulation results show that PTDMUS can improve the computation performance effectively, while keeping good precision and recall of result.

In the future, we are going to apply PTDMUS to mobile edge computing frameworks for making the multi-criteria decision on the dynamic placement of drone base stations, thus providing reliable communication services for specific purposes and scenarios.

%\appendices

%\section{}

\ifCLASSOPTIONcompsoc
  \section*{Acknowledgments}
\else
  \section*{Acknowledgment}
\fi

This research is partially supported by Ministry of Science and Technology under the Grant MOST 107-2221-E-027-099-MY2 and MOST 108-2634-F-009-006- through Pervasive Artificial Intelligence Research (PAIR) Labs, Taiwan.

\ifCLASSOPTIONcaptionsoff
  \newpage
\fi



\bibliographystyle{IEEEtran}
%\bibliographystyle{ieeetr}
%\vspace{-10pt}
\bibpreamble
\renewcommand{\bibfont}{\footnotesize}
\bibliography{IEEEabrv,reference}
%\begin{thebibliography}{1}


%\end{thebibliography}


%\vspace{-1.2cm}
\begin{IEEEbiography}[{\includegraphics[width=1in,height=1.25in,clip,keepaspectratio]{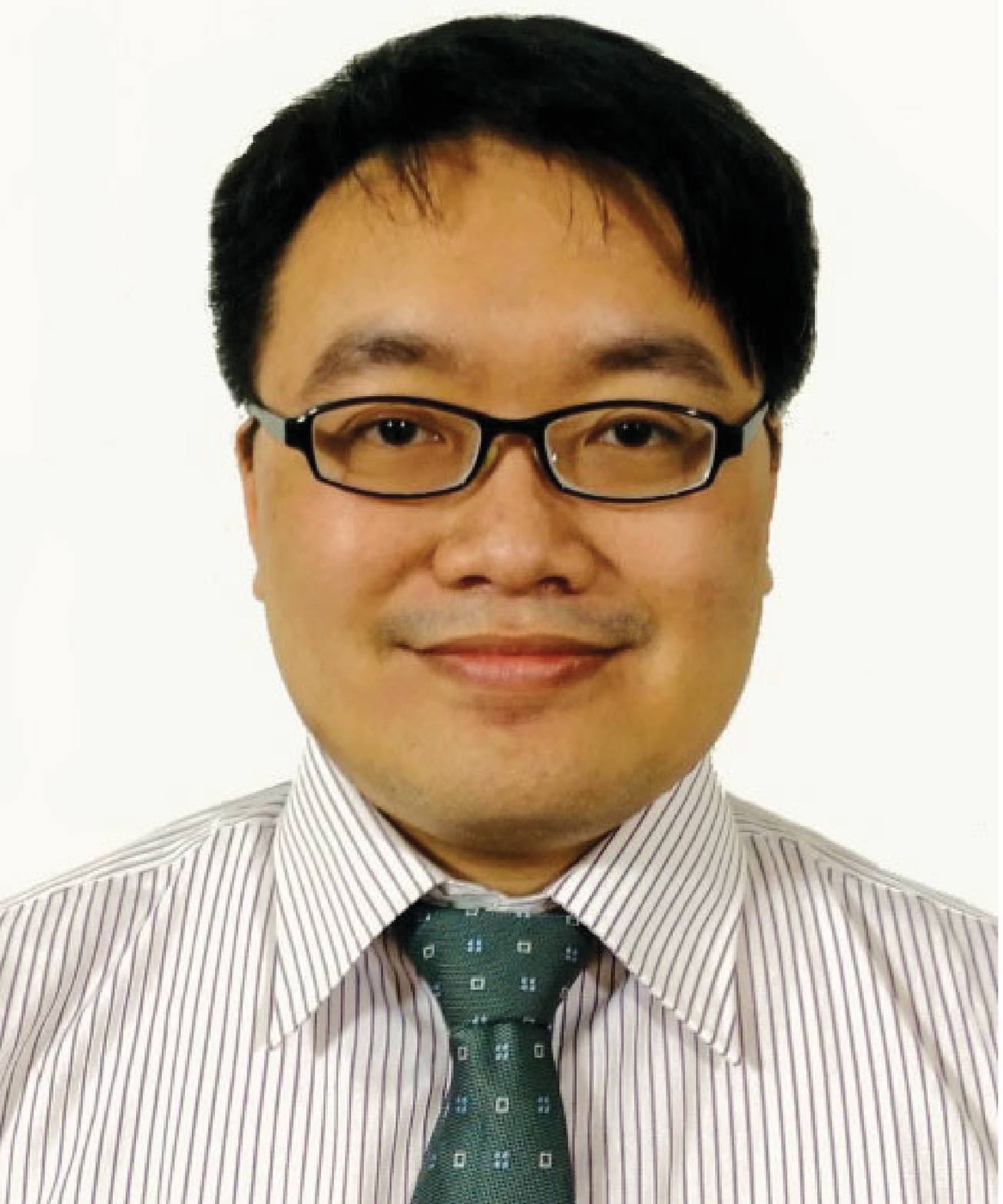}}]{Chuan-Chi Lai}
	%\begin{IEEEbiography}{Chuan-Chi Lai}
	is currently holding a post-doctoral position in the Department of Electrical and Computer Engineering at National Chiao Tung University, Taiwan, R.O.C. He received his Ph. D. in Computer Science and Information Engineering from National Taipei University of Technology (Taipei Tech), Taiwan in 2017. He won Excellent Paper Award and Best Paper Award in ICUFN 2015 and WOCC 2018 conferences, respectively. His current research interests are in the areas of data management and dissemination techniques in mobile wireless environments, mobile ad-hoc and sensor networks, distributed query processing over moving objects, and analysis and design of distributed algorithms.
\end{IEEEbiography}
%\vspace{-30pt}

\begin{IEEEbiography}[{\includegraphics[width=1in,height=1.25in,clip,keepaspectratio]{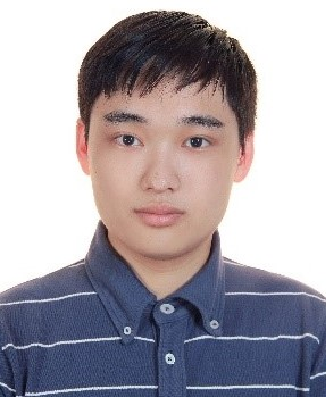}}]{Tien-Chun Wang}
	%\begin{IEEEbiography}{Yu-De Lin}
	received his MS degree in the Department of Computer Science and Information Engineering, National Taipei University of Technology (Taipei Tech), TAIWAN, R.O.C, in 2017. He joined Applied Computing Laboratory in 2015 and interested in developing searching algorithms for distributed systems. Now, he is an engineer in the department of mobile device development, Compal Electronics, Inc., Taiwan.
\end{IEEEbiography}
%\vspace{-30pt}

\begin{IEEEbiography}[{\includegraphics[width=1in,height=1.25in,clip,keepaspectratio]{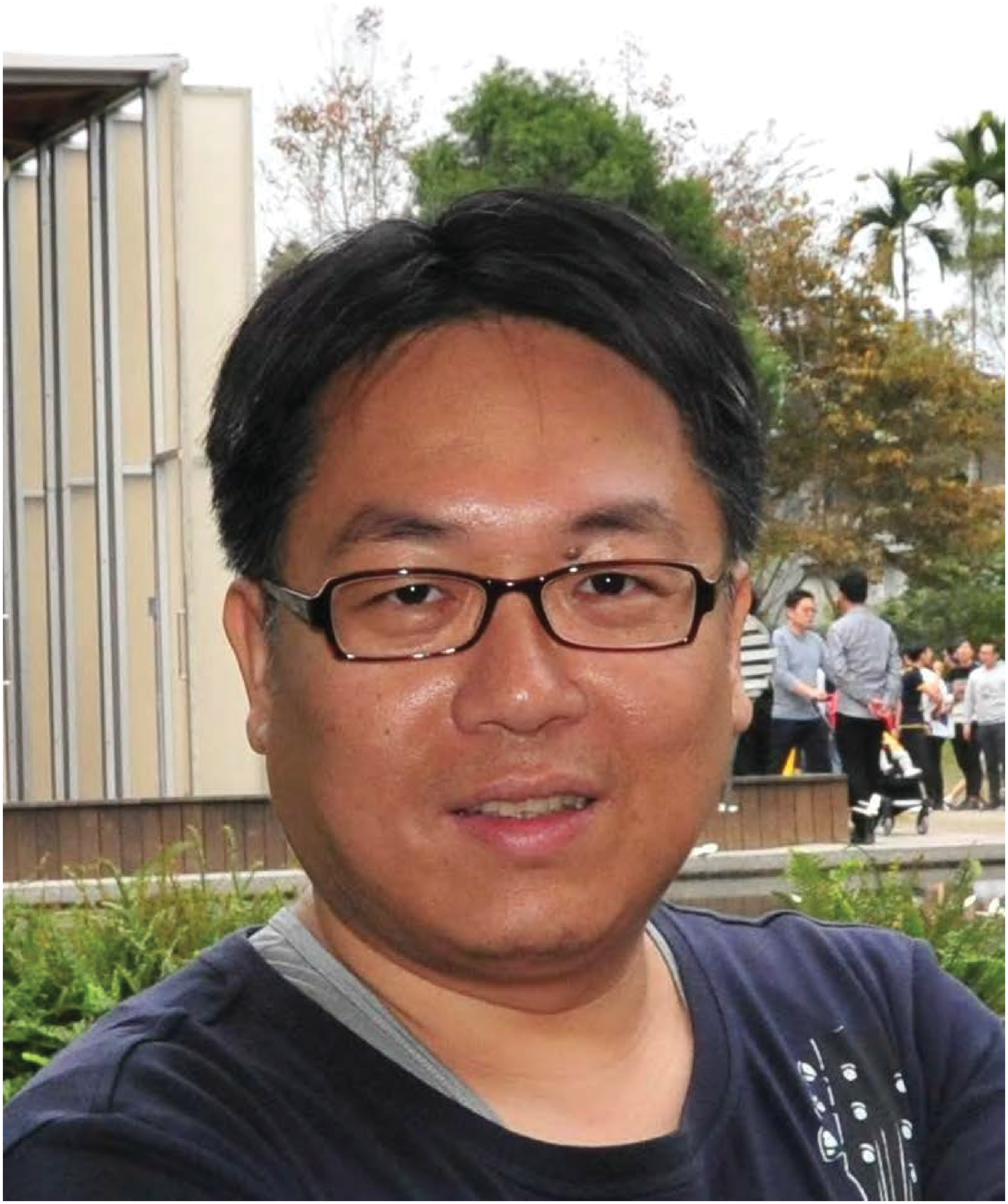}}]{Chuan-Ming Liu}
	%\begin{IEEEbiography}{Chuan-Ming Liu}
	(M'03) received the Ph.D. degree in computer science from Purdue University, West Lafayette, IN, USA, in 2002. 
	Dr. Liu is a professor in the Department of Computer Science and Information Engineering, National Taipei University of Technology (Taipei Tech), TAIWAN, R.O.C. In 2010 and 2011, he has held visiting appointments with Auburn University, Auburn, AL, USA, and the Beijing Institute of Technology, Beijing, China. He has services in many journals, conferences and societies as well as published more than 100 papers in many prestigious journals and international conferences. His current research interests include big data management and processing, uncertain data management, data science, spatial data processing, data streams, ad-hoc and sensor networks, and location based services.
\end{IEEEbiography}
%\vspace{-30pt}

%\newpage
\begin{IEEEbiography}[{\includegraphics[width=1in,height=1.25in,clip,keepaspectratio]{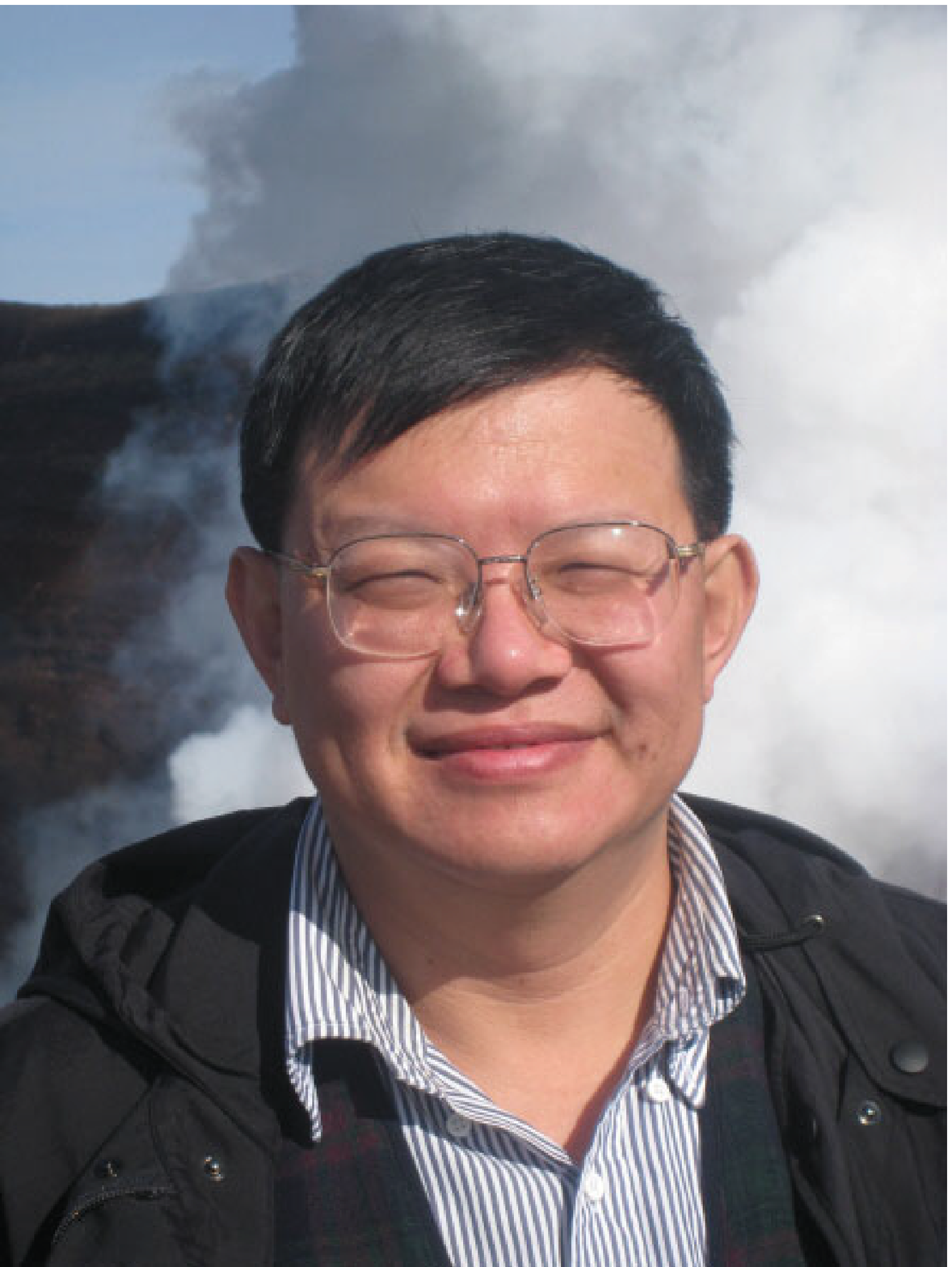}}]{Li-Chun Wang}
	(M'96 -- SM'06 -- F'11) received Ph. D. degree from the Georgia Institute of Technology, Atlanta, in 1996. From 1996 to 2000, he was with AT\&T Laboratories, where he was a Senior Technical Staff Member in the Wireless Communications Research Department. Since August 2000, he has joined the Department of Electrical and Computer Engineering of National Chiao Tung University in Taiwan and is jointly appointed by Department of Computer Science and Information Engineering of the same university.
	
	Dr. Wang was elected to the IEEE Fellow in 2011 for his contributions to cellular architectures and radio resource management in wireless networks. He won two Distinguished Research Awards of National Science Council, Taiwan in 2012 and 2017, respectively. He was the co-recipients of IEEE Communications Society Asia-Pacific Board Best Award (2015), Y. Z. Hsu Scientific Paper Award (2013), and IEEE Jack Neubauer Best Paper Award (1997).
	
	His current research interests are in the areas of software-defined mobile networks, heterogeneous networks, and data-driven intelligent wireless communications. He holds 19 US patents, and has published over 200 journal and conference papers, and co-edited a book, "Key Technologies for 5G Wireless Systems," (Cambridge University Press 2017).
\end{IEEEbiography}


%\vfill


\end{document}